\documentclass[english,11pt]{article}

\usepackage{amssymb}
\usepackage{natbib}
\usepackage{graphicx}
\usepackage{enumitem}
\usepackage{amsmath,amssymb,amsthm}
\newcommand{\calC}{{\mathcal{C}}}
\newcommand{\tran}{^\intercal}
\newcommand{\R}{\mathbb{R}}
\newcommand{\convdis}{\stackrel{d}{\Rightarrow}}
\newcommand{\calI}{{\mathcal{I}}}
\newcommand{\sign}{\operatorname{sign}}
\newcommand{\diag}{\operatorname{diag}}
\newcommand{\argmin}{\operatorname{argmin}}
\usepackage[ruled,vlined]{algorithm2e}
\usepackage{geometry}
\usepackage{algorithmic, setspace}
\renewcommand{\algorithmicrequire}{\textbf{STEP}}

\newtheorem{assumption}{Assumption}
\newcommand{\EE}[2][]{\mathbb{E}_{#1}\left[#2\right]}
\newcommand{\N}{\mathcal{N}}
\usepackage{subcaption}
\newcommand{\convp}{\stackrel{p}{\rightarrow}}
\newcommand{\Var}[2][]{\operatorname{Var}_{#1}\left[#2\right]}
\newcommand{\Cov}[2][]{\operatorname{Cov}_{#1}\left[#2\right]}
\newcommand{\calE}{{\mathcal{E}}}
\newcommand{\Rom}[1]{\uppercase\expandafter{\romannumeral #1\relax}}
\newcommand{\rom}[1]{\lowercase\expandafter{\romannumeral #1\relax}}
\newcommand{\bbT}{{\mathbb{T}}}
\newcommand{\bbQ}{{\mathbb{Q}}}
\newcommand{\bfr}{{\mathbf{r}}}
\newcommand{\Exp}[1]{\operatorname{exp}\left[#1\right]}
\newcommand{\bfs}{{\mathbf{s}}}
\newcommand{\bbL}{{\mathbb{L}}}
\newcommand{\bbM}{{\mathbb{M}}}
\newcommand{\bbD}{{\mathbb{D}}}

\newtheorem{theorem}{Theorem}
\newtheorem{lemma}[theorem]{Lemma} 
\newtheorem{proposition}[theorem]{Proposition} 
\newtheorem{remark}[theorem]{Remark}

\usepackage[colorlinks=true,linkcolor=blue,citecolor=blue,urlcolor=blue,breaklinks]{hyperref}

\makeatletter
\newcommand*\if@single[3]{%
  \setbox0\hbox{${\mathaccent"0362{#1}}^H$}%
  \setbox2\hbox{${\mathaccent"0362{\kern0pt#1}}^H$}%
  \ifdim\ht0=\ht2 #3\else #2\fi
  }
\newcommand*\rel@kern[1]{\kern#1\dimexpr\macc@kerna}
\newcommand*\widebar[1]{\@ifnextchar^{{\wide@bar{#1}{0}}}{\wide@bar{#1}{1}}}
\newcommand*\wide@bar[2]{\if@single{#1}{\wide@bar@{#1}{#2}{1}}{\wide@bar@{#1}{#2}{2}}}
\newcommand*\wide@bar@[3]{%
  \begingroup
  \def\mathaccent##1##2{%
    \if#32 \let\macc@nucleus\first@char \fi
    \setbox\z@\hbox{$\macc@style{\macc@nucleus}_{}$}%
    \setbox\tw@\hbox{$\macc@style{\macc@nucleus}{}_{}$}%
    \dimen@\wd\tw@
    \advance\dimen@-\wd\z@
    \divide\dimen@ 3
    \@tempdima\wd\tw@
    \advance\@tempdima-\scriptspace
    \divide\@tempdima 10
    \advance\dimen@-\@tempdima
    \ifdim\dimen@>\z@ \dimen@0pt\fi
    \rel@kern{0.6}\kern-\dimen@
    \if#31
      \overline{\rel@kern{-0.6}\kern\dimen@\macc@nucleus\rel@kern{0.4}\kern\dimen@}%
      \advance\dimen@0.4\dimexpr\macc@kerna
      \let\final@kern#2%
      \ifdim\dimen@<\z@ \let\final@kern1\fi
      \if\final@kern1 \kern-\dimen@\fi
    \else
      \overline{\rel@kern{-0.6}\kern\dimen@#1}%
    \fi
  }%
  \macc@depth\@ne
  \let\math@bgroup\@empty \let\math@egroup\macc@set@skewchar
  \mathsurround\z@ \frozen@everymath{\mathgroup\macc@group\relax}%
  \macc@set@skewchar\relax
  \let\mathaccentV\macc@nested@a
  \if#31
    \macc@nested@a\relax111{#1}%
  \else
    \def\gobble@till@marker##1\endmarker{}%
    \futurelet\first@char\gobble@till@marker#1\endmarker
    \ifcat\noexpand\first@char A\else
      \def\first@char{}%
    \fi
    \macc@nested@a\relax111{\first@char}%
  \fi
  \endgroup
}
\makeatother

\author{Sifan Liu$^{1}$, Snigdha Panigrahi$^{2}$ \\  $^{1}$Department of Statistics, Stanford University\\  $^{2}$Department of Statistics, University of Michigan}

\date{}  

\begin{document}

\title{Selective Inference with Distributed Data}

\maketitle

\begin{abstract}
As datasets grow larger, they are often distributed across multiple machines that compute in parallel and communicate with a central machine through short messages. In this paper, we focus on sparse regression and propose a new procedure for conducting selective inference with distributed data. Although many distributed procedures exist for point estimation in the sparse setting, few options are available for estimating uncertainties or conducting hypothesis tests based on the estimated sparsity. We solve a generalized linear regression on each machine, which then communicates a selected set of predictors to the central machine. The central machine uses these selected predictors to form a generalized linear model (GLM). To conduct inference in the selected GLM, our proposed procedure bases approximately-valid selective inference on an asymptotic likelihood. The proposal seeks only aggregated information, in relatively few dimensions, from each machine which is merged at the central machine for selective inference. By reusing low-dimensional summary statistics from local machines, our procedure achieves higher power while keeping the communication cost low. This method is also applicable as a solution to the notorious p-value lottery problem that arises when model selection is repeated on random splits of data. 
\end{abstract}

\section{Introduction}
\label{sec:1}

In the past few years, it has become increasingly important to be able to solve problems with a large number of training samples and predictors. 
Typically, such big datasets cannot be stored or analyzed on a single machine.
Distributed frameworks are widely used when dealing with big datasets spread over multiple machines \citep{bekkerman2011scaling,bertsekas2015parallel}.
One of the simplest and most popular approaches in this framework is  ``divide-and-conquer", which is also known as ``split-and-merge" or ``one-shot" approach; 
see for example the early work by \citet{mcdonald2009efficient,zinkevich2010parallelized,zhang2012communication}.
Most of these approaches use only one round of communication.
Each local machine estimates the unknown parameter, using its subset of the training data, and communicates the estimator to a central machine which merges the local estimators to obtain a global estimator.

In our paper, we focus on a sparse regression setting where only a few of the measured predictors affect the response.
The goal in this setup is usually two-fold: (i) select relevant predictors, and model the response by using the estimated sparsity, (ii) provide uncertainties or conduct hypothesis tests for the selected regression parameters, and all this while, respect the distributed nature of data. 
Substantial progress has been made on the first goal for sparse problems.
For example, \citet{lee2015communication} average locally computed, debiased Lasso estimators, and show that the averaged estimator achieves the same estimation rate as the full-sample Lasso, as long as the number of machines is not too large.
\citet{chen2014split} prove that the models aggregated via majority voting, based on variables selected by local machines, are consistent under some conditions.
More recently, \citet{battey2018distributed} provide an approach for conducting hypothesis testing in the distributed setting. However, these methods are limited to models that are fixed before the selection of relevant predictors, and do not offer inference in models that are formed only after selection with distributed data. As a result, the communication cost in prior work scales with the number of original predictors, which can be unnecessarily large in sparse settings. 


We introduce a new procedure for selective inference with distributed data.
Selective inference is a rigorous approach that accounts for the fact that the same data, used to select models, is re-used when providing confidence intervals and p-values.
Several ingenious tools have been developed to provide selective inference in sparse regression problems; please see papers by \cite{benjamini2005false, berk2013valid, belloni2015uniform, lee2016exact, tian2018selective, charkhi2018asymptotic, bachoc2019valid, panigrahi2021integrative}.
In this paper, our procedure re-uses data from all machines to base approximately-valid selective inference on an asymptotic ``selective likelihood". 
Having identified relevant predictors at different machines, we describe the relationship between our response and the predictors through a generalized linear model (GLM).
Our procedure only requires some aggregated information from each machine to deliver selective inference in a GLM with selected predictors.
For this reason, the developed techniques are also applicable to settings when datasets are distributed across different sites due to security, privacy, or ethical concerns, as encountered in the areas of differential privacy \citep{balcan2012distributed}, and federated learning \citep{mcmahan2017communication}.
More precisely, these sites can now merge aggregated information—without having to share their individual data—to infer in the selected GLM. 
The communication cost of our inferential procedure is only linear in the dimension of the selected model, which is relatively smaller than the initial dimension of the problem.
Finally, our procedure can be easily adapted to address the ``p-value lottery" problem that arises with model selection on random splits of data. 
The proposal serves as an efficient alternative to multi-splitting in  \cite{dezeure2015high}, and multi-carving in  \cite{schultheiss2021multicarving}.
Multi-carving is a more powerful version of multi-splitting, and is based on techniques that are known as carving \citep{fithian2014optimal, panigrahi2018carving}.
Closely aligned with the conceptual framework of multi-carving, our asymptotic selective likelihood uses the entire dataset more efficiently than multi-spitting.
At the same time, our procedure is significantly faster than  existing implementations of multi-carving without recourse to Markov chain Monte Carlo (MCMC) sampling.

The remaining paper is structured as follows.
We provide a slightly more technical account of our contributions after outlining the problem setup, and review related work in Section \ref{sec:2}.
In Section \ref{sec:3}, we describe our procedure for selective inference with distributed data.
In Section \ref{sec:4}, we provide an asymptotic justification for our selective likelihood in a GLM with selected predictors.
We discuss an adaptation of our procedure to solve the p-value lottery problem in Section \ref{sec:5}.
Section \ref{sec: experiments} illustrates an application of our method on simulated datasets and on a publicly available, medical dataset on intensive care unit (ICU) admissions.
We conclude the paper with a discussion in Section \ref{sec: conclusion}.
Proofs of our technical results are collected in the Appendix.

\section{Problem setup and background}
\label{sec:2}

In this section, we describe the distributed setup and introduce some background on selective inference with a single machine. Other related work is summarized at the end.
\subsection{Setup} 

 We consider a distributed setup with $K$ local machines, all connected to a central machine, referred to as machine 0.
Suppose that we observe $n$ i.i.d. observations 
$$(y_i, x_{i,1}, \ldots, x_{i,p})\in \R^{p+1},\  i\in [n],$$
where $[n]=\{1,2,\ldots, n\}$ for $n\in \mathbb{N}$.
Let $\calC^{(k)}\subset[n]$ denote the index set of the samples stored at machine $k$ for $0\leq k\leq K$. The index sets $\calC^{(k)}$ are disjoint and form a partition of $[n]$. Let $n_k=|\calC^{(k)}|$ be the cardinality of $\calC^{(k)}$. Let $\rho_k = \dfrac{n_k}{n}$ be the proportion of samples processed by local machine $k$. Then 
\[n=\sum_{k=0}^K n_k,\quad 1=\sum_{k=0}^K \rho_k.
\]
In the matrix form, $Y^{(k)}$ represents a response vector with 
$\left\{y_i: i \in \calC^{(k)}\right\}$
as its $n_k$ entries.
Similarly, $X^{(k)}$ represents a predictor matrix with
$\left\{(x_{i,1}, \ldots, x_{i,p}): i \in \calC^{(k)}\right\}$
in its $n_k$ rows. 
Let $D^{(k)}=\left(Y^{(k)},X^{(k)}\right)$.


Suppose that the $K$ local machines solve (in parallel) a generalized linear regression with the Lasso penalty.
The loss function in our problem is derived from the log-likelihood of a distribution in the exponential family 
$$
f(y_i\mid \theta)= \exp\left( \frac{y_i \theta - A(\theta)}{\sigma^2}\right)\cdot c(y_i;\sigma),
$$
where the canonical mean parameter $\theta$ is linked to the $p$ predictors as 
$$\theta= x_i\tran\beta,$$
and $\sigma^2$ is the dispersion parameter. 
We assume that $\sigma$ is either known, or can be consistently estimated.
The loss function for machine $k$ is given by
\begin{align*}
\ell^{(k)}(\beta; D^{(k)}) = \frac{1}{\sqrt{n}\rho_k} \sum_{i\in\calC^{(k)}} \left\{A(x_i\tran\beta)-y_ix_i\tran\beta \right\}.
\end{align*}
Let $\Lambda^{(k)} = \text{diag}\left(\lambda^{(k)}_1, \ldots, \lambda^{(k)}_p\right)$ be a diagonal matrix of positive regularization parameters.
Machine $k$ solves
\begin{equation}
\underset{\beta \in \R^p}{\text{minimize}} \; \ell^{(k)}(\beta; D^{(k)})+  \|\Lambda^{(k)}\beta\|_1.
\label{GL:regression}
\end{equation}
To simplify notations, we write $\Lambda^{(k)}= \Lambda$ for all $k\in [K]$.
But our procedure also works when $\Lambda^{(k)}$ are different across machines.
Denote by $\widehat\beta^{\Lambda,(k)}$ the Lasso estimator (the solution of problem~\eqref{GL:regression}), and 
denote by 
$$\widehat{E}^{(k)}=\left\{j\in [p]: \widehat\beta^{\Lambda,(k)}_j\neq 0\right\}$$
the indices of the selected predictors.
We use the symbol $E^{(k)}$ to denote the value of $\widehat{E}^{(k)}$ which we observe for our specific data $D^{(k)}$.
Fix $|E^{(k)}|=d^{(k)}$.

\subsection{Problem}

As described above, the $K$ local machines return as output the sets $\{E^{(k)},k\in[K]\}$, which are communicated to the central machine. 
The central machine aggregates the selected sets of predictors as
\begin{equation}
\label{aggregate}
\widehat{E}=\texttt{Aggregate}\left(\left\{\widehat{E}^{(k)}, k\in [K]\right\}\right).
\end{equation}
Suppose, a local machine $k'$ (or its output) was not used while forming the aggregated set $\widehat{E}$.
In this case, we can simply combine $D^{(k')}$ with $D^{(0)}$, and proceed with our prescription.
For now, we focus on aggregation rules that satisfy
$\widehat{E} \supseteq \cup_{k\in[K]}\widehat{E}_k$.
As a concrete example, we may consider 
$$\widehat{E}=\underset{k\in [K]}{\bigcup} \widehat{E}^{(k)}.$$
Later in Appendix~\ref{sec: general aggregation}, we describe the method for general aggregation rules with a slight modification.

Consistent with our notations, we let $E$ be the observed value of $\widehat{E}$, and let $|E|=d$.
The central machine models our real-valued response as a GLM with the following density
\begin{equation}
f(y\mid x_{E},\beta_E) = \exp\left( \frac{y x_{E}\tran\beta_E - A(x_{E}\tran\beta_E)}{\sigma^2}\right)\cdot c(y;\sigma).
\label{sel:model}
\end{equation}
Here, for a vector $x$ and $E\subset[p]$, $x_E$ denotes the subvector of $x$ with indices in $E$. 
Similarly, for a matrix $X$, $X_E$ consists of the columns of $X$ with indices in $E$.

Equivalently, we may summarize the modeling workflow in our paper as follows. 
Having selected a subset of important predictors, each local machine communicates a base GLM to the central machine.
The central machine combines these base models, through $E$, to form the selected GLM in \eqref{sel:model}.

Some key questions arise when we seek selective inference in the selected GLM.
\begin{enumerate}
\item Can the central machine re-use data from the local machines to deliver selective inference?
Of course, naive inference, which uses all the data without adjusting for selection bias in the selected GLM, falls short of coverage guarantees, as illustrated on one of our simulated instances in Figure~\ref{fig: naive}. (See Section~\ref{sec: experiments} for details of this simulation.)
\item Selective inference must respect the distributed nature of data, as done at the time of selection. 
What information does the central machine seek from the local machines for selective inference?
How many exchanges of communication does it take to compute selective inference with distributed data?
\end{enumerate}

\begin{figure}
\centering
\includegraphics[width=.5\textwidth]{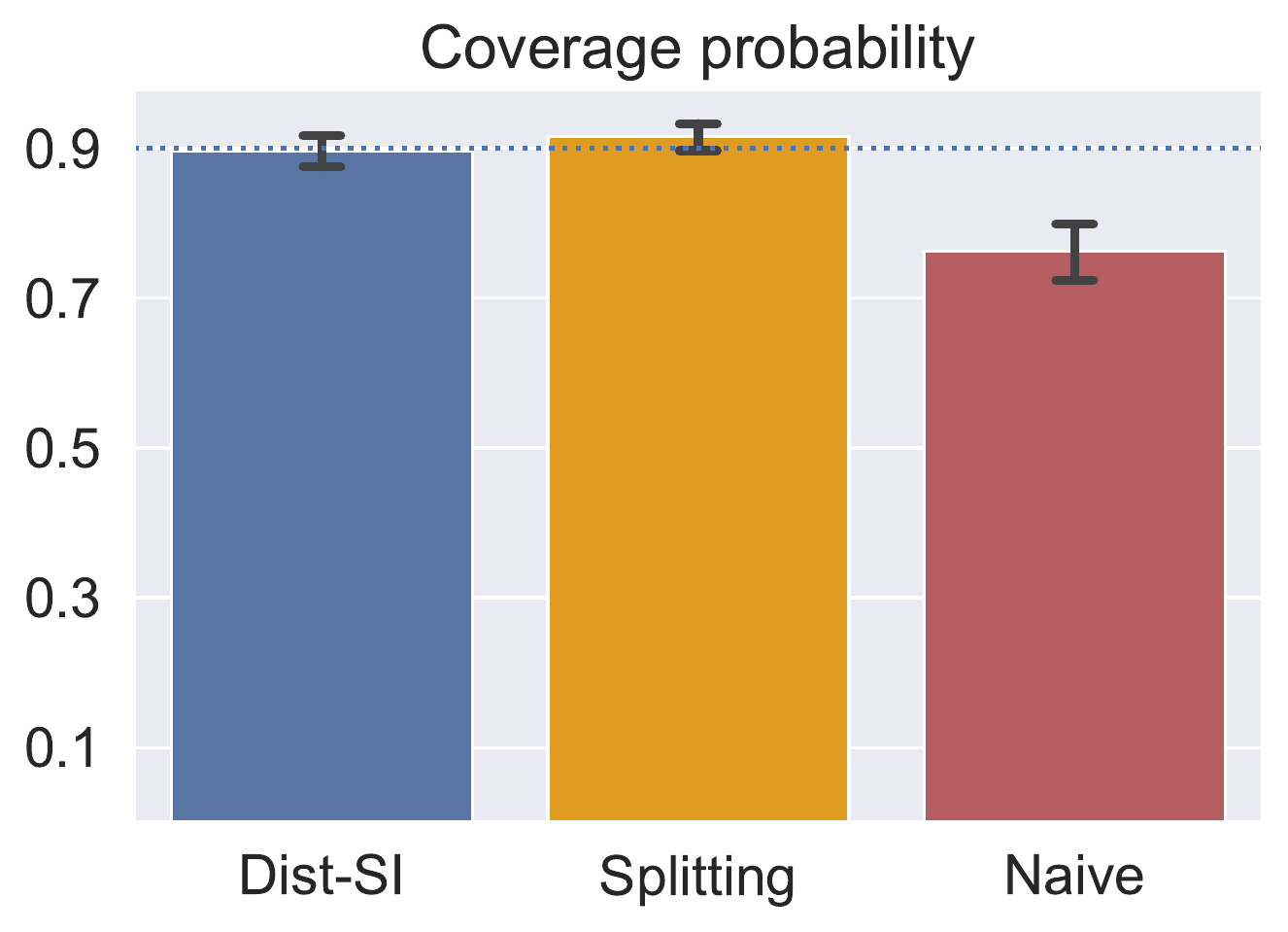}
\caption{Coverage probabilities of ``Dist-SI" (our procedure, distributed selective inference), ``Splitting", and ``Na\"{i}ve" methods for a Gaussian linear model, with $2$ local machines and a central machine. The $2$ local machines and the central machine each have $1000$ samples. The pre-specified level of coverage is $0.9$ which is indicated by the dotted horizontal line at $0.9$.}
\label{fig: naive}
\end{figure}

\subsection{Some background}
We provide some background in a rather simple setup with $K=1$, and consider the special case of a linear model. Machine $1$ solves problem~\eqref{GL:regression} with penalty $\Lambda^{(1)}=\lambda I_p$
on the subset $D^{(1)}$.
After observing $\widehat{E}=E$, each observation in our response vector is modeled independently as
\begin{equation*}
y \sim \mathcal{N}( x_E\tran \beta_E, \sigma^2).
\end{equation*}
This linear model is the GLM \eqref{sel:model} with
$$A(x_E\tran \beta_E)= \frac{1}{2} (x_E\tran \beta_E)^2, \; c(y;\sigma)= -\frac{y^2}{2\sigma^2} .$$
Can we re-use data from machine $1$, and combine it with $D^{(0)}$ to form inference in the selected linear model?
Recent work in selective inference solves an asymptotically related problem by forming a ``selective likelihood",  which we review very briefly below.

First, we note that the Gaussian regression in \eqref{GL:regression} can be re-written as
\begin{equation}
\underset{\beta \in \mathbb{R}^p}{\text{minimize}} \; \frac{1}{\sqrt{n}} \sum_{i\in [n]}  \frac{1}{2}(y_i - x_i\tran\beta)^2+ \lambda \|\beta\|_1-\sqrt{n}\omega_n\tran\beta,
\label{randomized:Lasso}
\end{equation}
where 
\begin{align*}
\omega_n = \frac{1}{n_k} \sum_{i\in\calC^{(1)}} \frac{1}{2}x_i(y_i - x_i\tran \widehat\beta^{\lambda,(1)})-\frac{1}{n} \sum_{i\in [n]}  \frac{1}{2}x_i(y_i - x_i\tran \widehat\beta^{\lambda,(1)}).
\label{equ: def omega}
\end{align*}
Regression of the form \eqref{randomized:Lasso}, with an added randomization variable $\omega_n$, is called the randomized Lasso in \cite{harris2016selective}, and is motivated by the randomized-response approach in \cite{tian2018selective}.
\citet{panigrahi2021integrative} show that 
$$\sqrt{n}\omega_n \convdis \mathcal{N}(0_p, \Sigma_\Omega)$$ where
$\Sigma_\Omega= \dfrac{\rho_0 \sigma^2}{\rho_1}\mathbb{E}[x_i x_i\tran],$
and $\omega_n$ is independent of other data variables involved during selection.
By casting selection into the randomized framework, a ``selective likelihood" can be obtained by conditioning the usual likelihood on a subset of the selection event
$$\left\{\widehat{E}=E\right\}.$$
In principle, conditioning on the above event or a subset of the same event yields us valid selective inference, and simultaneously permits us to re-use data from the selection steps.

Suppose the $n$ observations in our response vector are drawn as independent Gaussian variables with the same variance, and say, we solve \eqref{randomized:Lasso} with $\sqrt{n}\omega_n \sim \mathcal{N}(0_p, \Sigma_\Omega)$.
A recent procedure by \cite{panigrahi2022approximate} uses the selective likelihood to construct inference in the selected linear model: it centers interval estimates around the maximum likelihood estimator (MLE) of the selective likelihood, and estimates its variance by using the observed Fisher information matrix.
A major obstruction to make the procedure operational is the intractability of the selective likelihood.
For Gaussian data, the afore-mentioned paper provides tractable, approximately-valid selective inference by solving a convex optimization problem which yields a direct approximation to the selective likelihood function. 
Let $\widehat{\beta}^{(S)}_{E}$ and $\widehat{I}^{(S)}_{E,E}$ be the selective MLE and the observed Fisher information matrix respectively, derived from the score and curvature of the approximate likelihood. 
For a prespecified significance level $\alpha$, a two-sided $100\cdot(1-\alpha)\%$ confidence interval is now given by
$$\widehat{\beta}^{(S)}_{E,j} \pm z_{1-\alpha/2}\cdot  \frac{\widehat\sigma_j^{(S)}}{\sqrt{n}} \text{ for } j\in [E], $$
where $$\widehat\sigma_j^{(S)}=\sqrt{ \left(\widehat{I}^{(S)}_{E,E}\right)^{-1}_{j,j} }$$
is the estimated variance of the $j^{\text{th}}$ entry of $\sqrt{n}\widehat\beta^{(S)}_{E}$, and $z_{1-\alpha/2}$ is the $(1-\alpha/2)$-th quantile of a standard normal distribution.
The procedure closely resembles classical inference via maximum likelihood, except that the standard estimators are replaced by their selective (selection-corrected) analogs.

\subsection{Contributions and other related work}

Our paper develops a new procedure to deliver approximately-valid selective inference with distributed data.
We make three main contributions.
Re-using data from machines through a conditional approach is challenging in the distributed setup.
This is because a conditional approach proceeds by deriving an explicit representation of selection which is unavailable for distributed data. 
As a first contribution, we identify a simple representation for selection by developing a randomized framework in our problem.
Even with this representation, adopting a conditional approach is not immediate.
As reviewed in the preceding discussion, approximately-valid selective inference has been developed only for linear Gaussian models.
Not only does our paper provide an asymptotic selective likelihood for distributed data, but also establishes consistency of the likelihood function for a large class of generalized linear regression problems.
Second, our algorithm for selective inference simply requires aggregated information from each machine.
To infer under the selected GLM in \eqref{sel:model}, the central machine solves a straightforward, convex optimization after merging aggregated information from the local machines.
As a result, our techniques apply in situations when direct data-sharing between sites is not possible, or communication across sites is expensive.
Third, p-values computed under our asymptotic selective likelihood can be easily adapted to address the p-value lottery problem.
In this sense, our procedure is related to the multi-carving approach for improved replication \citep{schultheiss2021multicarving}, but yields us a much faster, sampling-free solution to this problem.

We conclude this section with some more related work. 
Within the distributed setting, much work has been devoted to the averaged M-estimator \citep{mcdonald2009efficient,zinkevich2010parallelized,zhang2012communication,rosenblatt2016optimality}. 
\citet{rosenblatt2016optimality} show that the averaged M-estimator is first-order equivalent to the centralized M-estimator in the fixed-dimension setting. \citet{dobriban2021distributed} study the efficiency of an estimator based on weighted average in the linear regression setting, as dimensions grow with sample sizes.
Some methods have taken a likelihood-centric approach, e.g., \citet{jordan2016} propose a surrogate likelihood where higher-order derivatives in a Taylor-series expansion of the full log-likelihood are replaced by local approximations.
\citet{lin2011aggregated} propose an aggregated estimator for generalized linear models (GLM), where the locally computed MLE and Hessian of the likelihood are merged for efficiency gains.
In work by \cite{huang2005sampling,neiswanger2013asymptotically,wang2013parallelizing,scott2016bayes,minsker2017robust,srivastava2018scalable}, distributed MCMC algorithms combine local posterior samples to obtain a global posterior distribution.

In the selective inference literature, a selective likelihood was appended to priors for Bayesian inference post selection in \cite{panigrahi2018scalable, panigrahi2021integrative}. 
The focus in these settings was on a category of variable selection rules that can be written as a set of polyhedral constraints on data. 
For the same category of selection rules, a separate section of papers \citep{lee2016exact, hyun2018exact, le2022more} construct an exact pivot to form confidence intervals and p-values. 
We, however, note that exact selective inference is only available for Gaussian models.
Moving beyond Gaussian models, \cite{taylor2018post} provide an asymptotic scheme to base selective inference on a GLM, within the usual regression context. For blackbox selection, \citet{liu2022black} propose to learn the selective likelihood by learning from the selection events of bootstrapped datasets.
We take a different approach in the distributed setup by casting the problem into a randomized framework, and provide an asymptotic likelihood function for the selected GLM.
A randomized framework for selective inference has been considered for better power in \cite{harris2016selective, panigrahi2017mcmc}, for a more efficient use of data in \cite{rasines2021splitting, panigrahi2022treatment}, and for stability in \cite{zrnic2020post}.

\section{Selective inference with distributed data}
\label{sec:3}

Before proceeding further, we fix some notations.
Recall, $|E^{(k)}|=d^{(k)}$, and $|E|=d$.
Let $\bar d=\underset{k\in [K]}{\sum}d^{(k)}$. 
For $k\in [K]$, the vector $\widehat{B}^{(k)}\in \R^{d^{(k)}}$ collects the nonzero components of the Lasso estimator $\widehat\beta^{\Lambda,(k)}$, and $\widehat{S}^{(k)}=\sign(\widehat{B}^{(k)})\in \R^{d^{(k)}}$ is the vector of signs for predictors that were selected by machine $k$.
Let $\widehat{Z}^{(k)}\in\R^{p - d^{(k)}}$ be the subgradient of the Lasso penalty for the inactive predictors, at the solution of the Lasso algorithm $\widehat\beta^{\Lambda,(k)}$.
As before, we will use the symbols $B^{(k)}$, $S^{(k)}$ and $Z^{(k)}$ for the observed values of $\widehat{B}^{(k)}$, $\widehat{S}^{(k)}$ and $\widehat{Z}^{(k)}$ respectively, and we use
$$\gamma^{(k)}= \Lambda\begin{pmatrix} S^{(k)} \\ Z^{(k)}\end{pmatrix}$$
to denote the $\R^p$-valued subgradient of the Lasso penalty at the solution.
Assuming that active predictors are stacked before the inactive ones in the gradient of the loss functions, we have
$$
-\nabla \ell^{(k)}(\widehat\beta^{\Lambda,(k)}; D^{(k)}) = \gamma^{(k)},
$$
and 
$$\|Z^{(k)}\|_\infty \leq 1, \quad \diag(S^{(k)})B^{(k)}>0, \text{ for } k\in [K].$$
We let
$$\widehat{B} =\begin{pmatrix} (\widehat{B}^{(1)})\tran & (\widehat{B}^{(2)})\tran & \ldots & (\widehat{B}^{(K)})\tran \end{pmatrix}\tran$$
Similarly, let the symbols $\widehat{S}$, $\widehat{Z}$ denote vectors that stack the corresponding quantities.
Let $B$, $S$, and $Z$ represent their observed values.

\subsection{Communication with central machine}

We begin by describing the two exchanges of communication, between the local machines and the central machine.

In exchange 1, the central machine sends the aggregated set of predictors $E$ to every local machine. 
After receiving $E$, local machine $k$ computes on $D^{(k)}$: 
\begin{enumerate}
\item[(i)] the standard MLE in the selected GLM as 
\begin{align*}
\widehat\beta_E^{(k)}=\underset{\beta\in\R^d}{\argmin} \frac1{\sqrt n \rho_k}\sum_{i\in\calC_k} A(x_{i,E}\tran\beta)-y_ix_{i,E}\tran\beta;
\end{align*}
\item[(ii)] the observed Fisher information (obs-FI) matrix at the MLE as
\[
\widehat\calI^{(k)}_{E,E}=\frac{1}{n_k} (X_E^{(k)})^{\intercal} \widehat{W}^k X_E^{(k)},\quad \text{where } \widehat{W}^{(k)}=\diag\left(\nabla^2 A(X^{(k)}_E \widehat\beta^{(k)}_E )\right).
\]
\end{enumerate}
Here, $\nabla^2 A(X^{(k)}_E \widehat\beta^{(k)}_E )$ denote the second derivative of $A$ evaluated at each coordinate of the vector $X^{(k)}_E \widehat\beta^{(k)}_E$.
Similarly, the central machine computes $\widehat\beta_E^{(0)}$ and $\widehat\calI^{(0)}_{E,E}$, using $D^{(0)}$.

In exchange 2, each local machine passes on these two quantities, $\widehat\beta_E^{(k)}$ and $\widehat\calI^{(k)}_{E,E}$, to the central machine. 
Suppose, $E$ were a fixed subset of predictors with no dependence on data.
It is well known in this situation that aggregating the standard MLE and observed Fisher information from each local machine delivers asymptotically valid inference; e.g., see the work by \cite{lin2011aggregated}.
Because the choice of model is data-dependent, each machine must now return some extra information aside from returning just the usual estimators.
Specifically, our procedure requires a part of the subgradient vector $\gamma^{(k)}$ after solving \eqref{GL:regression}, that is,
\begin{enumerate}
\item[] local machine $k$ sends $\gamma^{(k)}_{E}$ to the central machine, alongside the standard MLE and corresponding obs-FI matrix in the selected GLM.
\end{enumerate}
Relative to standard inference in the selected GLM, this extra information, per machine, does not come at any additional communication cost.
In fact, we note that the communication cost for selective inference is only $O(d^2)$ per machine.

\subsection{Merging information from local machines}
\label{sec: algorithm}

At the outset, the central machine forms the estimator
\begin{align}
\widehat\beta_E =\widehat\calI_{E,E}^{-1} \sum_{k\in  \{0\} \cup [K]}\rho_k \widehat\calI_{E,E}^{(k)} \widehat\beta^{(k)}_E.
\label{equ: aggregated MLE}
\end{align}
where
$$
\widehat\calI_{E,E}= \sum_{k\in  \{0\} \cup [K]} \rho_k \widehat\calI^{(k)}_{E,E}.
$$
The estimator in \eqref{equ: aggregated MLE}, proposed previously by \cite{lin2011aggregated}, merges the local MLE and the obs-FI matrix computed by each machine.
Given some regularity conditions, this estimator is asymptotically equivalent to the MLE using the full data, for a fixed $E\subset [p]$. 

In our problem, the central machine takes into account the data-dependent nature of our model by computing the selective MLE and the selective obs-FI that are selection-corrected values of $\widehat\beta_E$ and $
\widehat\calI_{E,E}$, respectively.
Both these values are derived from an asymptotic selective likelihood.
Deferring details of the selective likelihood and a theoretical justification of our procedure to the next section, we outline our algorithm for selective inference here.
We define some matrices for this purpose.

For an index set $E\subset[p]$, define $J_{E}\in\R^{|E|\times p}$ as the matrix that selects the elements in $E$, i.e., $J_E[i,j]$ is 1 if the $i$-th element of $E$ is the $j$-th element in $[p]$.
Let $g^{(j)}_k=J_{E^{(k)}} \gamma^{(j)} $ and $g^{(j)}=J_E \gamma^{(j)}$
collect components of $\gamma^{(k)}$, the subgradient from machine $k$, in the sets $E^{(j)}$ and $E$, respectively.
We compute the matrices $\widehat\Gamma$, $\widehat\Psi$, $\widehat\tau$, $\widehat\Theta$, $\widehat\Pi$, and $\widehat\kappa$ as follows.
The $(j,k)$ block of $\widehat\Gamma^{-1}$ is a $d^{(j)}\times d^{(k)}$ matrix given by
$$
\left\{\widehat\Gamma^{-1}\right\}_{j,k} = \begin{cases} 
      \left(\rho_k+\dfrac{\rho_k^2}{\rho_0}\right)\widehat\calI_{ E^{(k)},E^{(k)}} & \text{ if } j=k, \\[1.2em]
      \dfrac{\rho_j\rho_k}{\rho_0}\widehat\calI_{E^{(j)},E^{(k)}} & \text{ if } j\neq k
   \end{cases}.
$$
The $(k,1)$ block of $\widehat\Gamma^{-1}\widehat\Psi$ and $\widehat\Gamma^{-1}\widehat\tau$ are given by
\begin{align*}
\left\{\widehat\Gamma^{-1} \widehat\Psi\right\}_{k} &= \dfrac{\rho_k}{\rho_0} \widehat\calI_{E^{(k)},E}; \quad
\left\{\widehat\Gamma^{-1}\widehat\tau\right\}_{k}=-\rho_k g^{(k)}_{k} - \dfrac{\rho_k}{\rho_0} \sum_{j=1}^K \rho_j g^{(j)}_{k}.
\end{align*}
Finally, define $\widehat\Theta$, $\widehat\Pi$, $\widehat\kappa$ as 
\begin{align*}
\widehat\Theta^{-1} ={\frac{1}{\rho_0}} \widehat\calI_{E,E}-\widehat\Psi\tran \widehat\Gamma^{-1}\widehat\Psi;\quad \widehat\Theta^{-1}\widehat\Pi &= \widehat\calI_{E,E}; \quad \widehat\Theta^{-1} \widehat\kappa=\widehat\Psi\tran\widehat\Gamma^{-1}\widehat\tau+{\sum_{j=1}^K\frac{\rho_j}{\rho_0} g^{(j)} }. 
\end{align*}

The central machine delivers selective inference by solving the following interior-point algorithm
\begin{equation}
\label{optimization:inference}
\widehat{V}_{\widehat\beta_E}^{\star}= \underset{V\in\R^{\bar d}}{\argmin}\; \frac{1}{2}(\sqrt{n}V - \widehat\Psi \sqrt{n}\widehat\beta_E -\widehat\tau)\tran \widehat\Gamma^{-1}(\sqrt{n}V - \widehat\Psi \sqrt{n}\widehat\beta_E -\widehat\tau) + \text{Barr}_{\mathcal{O}_S}(\sqrt{n}V),
\end{equation}
where
$$\mathcal{O}_S=\left\{\left(V^{(1)}, V^{(2)},\ldots, V^{(K)}\right): V^{(k)}\in \mathbb{R}^{d^{(k)}}, \ \sign(V^{(k)})= S^{(k)} \text{ for } k\in [K]\right\}\subset \mathbb{R}^{\bar d},$$
and $\text{Barr}_{\mathcal{O}_S}(V)$ is a barrier penalty that takes the value $\infty$ if $V= \left(V^{(1)}, V^{(2)},\ldots, V^{(K)}\right) \notin \mathcal{O}_S$.
The selective MLE and the selective obs-FI are equal to
\begin{equation}
\label{selective:MLE}
\sqrt{n}\widehat\beta^{(S)}_E = \sqrt{n} \widehat\Pi^{-1}\widehat\beta_E - \widehat\Pi^{-1}\widehat\kappa + \widehat\calI_{E,E}^{-1}\widehat\Psi\tran \widehat\Theta^{-1}(\widehat\Psi \sqrt{n}\widehat\beta_E +\widehat\tau -\sqrt{n}\widehat{V}_{\widehat\beta_E}^{\star}),
\end{equation}
\begin{equation}
\label{obs:FI}
 \widehat\calI^{(S)}_{E,E} = \widehat\calI_{E,E}\left(\widehat\Theta^{-1}+ \widehat\Psi\tran \widehat\Gamma^{-1}\widehat\Psi - \widehat\Psi\tran \widehat\Gamma^{-1} \left(\widehat\Gamma^{-1}  + \nabla^2\text{Barr}_{\mathcal{O}_S}\left(\sqrt{n}\widehat{V}_{\widehat\beta_E}^{\star}\right)\right)^{-1}\widehat\Gamma^{-1} \widehat\Psi\right)^{-1}  \widehat\calI_{E,E},
\end{equation}
respectively. 
Algorithm~\ref{algo} summarizes our procedure for selective inference under the selected GLM in \eqref{sel:model}.

\begin{algorithm}
\setstretch{1.2}
\caption{Selective inference with Distributed Data}
\label{algo}
\SetKwInOut{Input}{Input}
\SetKwInOut{Output}{Output}
\vspace{2mm}

\algorithmicrequire{ \textbf{1:} Variable Selection at Local Machines}

Machine $k$ solves \eqref{GL:regression} and sends $E^{(k)} =\text{Support}(\widehat\beta^{\Lambda,(k)})$ to the central machine.

\vspace{2mm}
\algorithmicrequire{ \textbf{2:} Modeling with selected predictors}

Central Machine aggregates $E^{(k)}$ and forms the selected GLM in \eqref{sel:model}.

\vspace{2mm}
\algorithmicrequire{ \textbf{3:} Communication with Central Machine} 

 \hspace*{0.2cm} Exchange 1: Central machine sends the set $E$ to the local machines.

\hspace*{0.2cm}  Exchange 2: Local machine $k$ sends back the following information
\begin{align*}
&\text{local estimators: } \  \widehat\beta_{E}^{(k)}, \; \widehat\calI_{E,E}^{(k)};
&\text{subgradient at $\widehat\beta^{\Lambda,(k)}$: } \ \gamma^{(k)}_{E}
\end{align*}

\algorithmicrequire{ \textbf{4:} Selective Inference at Central Machine}

\begin{enumerate}[label=(\Alph*), leftmargin=1.1cm]
\item Compute the aggregated MLE $\widehat\beta_E$ defined by Equation~\eqref{equ: aggregated MLE}.

\item Solve the $\bar d$-dimensional convex optimization in \eqref{optimization:inference}.

\item Compute $\widehat\beta^{(S)}_E$ and  $\widehat\calI^{(S)}_{E,E}$ as stated in \eqref{selective:MLE} and \eqref{obs:FI}.
\end{enumerate}

\hspace*{0.2cm} 
Let
$$\widehat\sigma_j^{(S)}=\sqrt{ \left(\widehat{I}^{(S)}_{E,E}\right)^{-1}_{j,j} }$$
\vspace{2mm}
\hspace*{0.2cm} Compute two-sided p-values at level $\alpha$ as
$$ 2\cdot \text{min}\left( \Phi\left(\frac{\sqrt{n}}{\widehat\sigma_j^{(S)}}\widehat{\beta}^{(S)}_{E,j}\right), \bar\Phi\left(\frac{\sqrt{n}}{\widehat\sigma_j^{(S)}}\widehat{\beta}^{(S)}_{E,j}\right)\right)$$
where $\Phi=1-\bar\Phi$ is the CDF of the standard normal distribution.

\vspace{2mm}
\hspace*{0.2cm} Compute two-sided $100\cdot(1-\alpha)\%$ confidence intervals as
$$\widehat{\beta}^{(S)}_{E,j} \pm z_{1-\alpha/2} \cdot \frac{\widehat\sigma_j^{(S)}}{\sqrt{n}} \text{ for } j\in [E].$$
\end{algorithm}

\section{Theory}
\label{sec:4}

 We provide an asymptotic justification for our procedure in this section.
 In the remaining section, we let our parameter of interest be $\beta_E=\beta_{E,n}$ which is the population minimizer for the regression problem
$$
\underset{b}{\text{argmin}} \frac{1}{\sqrt{n}} \mathbb{E}\left[\sum_{i\in [n]} \left\{A(x_{i, E}\tran b)-y_ix_{i,E}\tran b \right\}\right].
$$

\subsection{A randomized representation of selection}
\label{sec: randomization}
We obtain our asymptotic selective likelihood from a conditional distribution of the aggregated MLE. 
As a first step in this direction, we develop a randomized framework which yields us a representation of selection with distributed data.

We use $\nabla A(X\beta)$ to denote the vector in $\R^n$ whose $i^{\text{th}}$ coordinate is the first derivative of $A(\cdot)$ at $x_i\tran\beta$.
Similar notations are used for higher derivatives of $A(\cdot)$.
Define the randomization variables 
\begin{equation}
\Omega=\left({\omega_n^{(1)}}\tran,\ldots,{\omega_n^{(K)}}\tran\right)\tran,
\label{randomization:distributed}
\end{equation}
where
\begin{align*}
\omega_n^{(k)}&= 
\frac{1}{ n} X\tran(\nabla A(X\widehat\beta^{\Lambda,(k)}) - Y) - \frac{1}{ n_k} X^{(k),\intercal}(\nabla A(X^{(k)} \widehat\beta^{\Lambda,(k)}) - Y^{(k)}),
\end{align*}
and $X$ and $Y$ are obtained by stacking $X^{(k)}$ and $Y^{(k)}$, for $k\in [K]$. 
We can re-formulate the generalized linear regression \eqref{GL:regression} for machine $k$ as
\begin{equation}
\underset{\beta \in \R^p}{\text{minimize}} \; \frac{1}{\sqrt{n}} \sum_{i\in [n]} \left\{A(x_i\tran\beta)-y_ix_i\tran\beta \right\}+  \|\Lambda\beta\|_1 -(\sqrt{n}{\omega^{(k)}_n})\tran\beta.
\label{GL:regression:randomized}
\end{equation}

Consider the following assumptions. 

\begin{assumption}
\label{assump: missed variables}
For $k\in [K]$, let $\widetilde{E}^{(k)}= E\setminus E^{(k)}$. 
For $j\in \widetilde{E}^{(k)}$, either $\beta_{j,n}=O(n^{-1/2})$ or $X_{E^{(k)}}\calI_{E^{(k)},E^{(k)}}^{-1} \calI_{E^{(k)},j}  = X_j$.
\end{assumption}

\begin{assumption}
\label{assump: glm regularity}
The aggregated MLE \eqref{equ: aggregated MLE}, in the selected GLM, can be written as
\begin{align*}
\sqrt n(\widehat\beta_E-\beta_{E,n})=-\calI_{E,E}^{-1} \nabla\ell(\beta_E)+o_p(1),
\end{align*}
where $\calI_{E,E}$ is the Fisher information in the same model. 
\end{assumption}

Assumption \ref{assump: missed variables} states conditions on predictors that are present in the selected GLM, but are not selected by machine $k$. 
The conditions imply that our asymptotic assertions hold as long as such predictors are either weak in strength, or have a high partial correlation with a predictor in the selected set $E^{(k)}$. 
The regularity condition in Assumption \ref{assump: glm regularity} states that the  aggregated MLE admits an asymptotically linear representation.
This condition is satisfied when the standard MLE, based on the full data, admits the same linear representation, and when the aggregated MLE is asymptotically equivalent to the standard MLE.
The latter fact has been shown to hold under some regularity conditions in \cite{lin2011aggregated}.

The next Theorem \ref{thm: normality of omega} provides the asymptotic distribution for the randomization variables in \eqref{randomization:distributed}, and Proposition \ref{prop: asymp indep} finds their joint distribution with other variables involved in selection.

\begin{theorem}
\label{thm: normality of omega}
Let $U=\diag(\rho_1^{-1},\ldots,\rho_{K}^{-1})-\mathbf{1}_{K\times K}$, and let $W=\diag(\nabla^2 A(X_E\beta_{E,n} ) )$.
Suppose
\begin{align*}
\calI=\EE{\frac1n X\tran WX }
\end{align*}
is the full Fisher information matrix at $\beta_{E,n}$.
Define
\begin{align*}
\Sigma_{\Omega}=U\otimes \calI,
\end{align*}
the Kronecker product of $U$ and $\calI$.
We have
\begin{align*}
\sqrt n\, \Omega \convdis \N_{pK}\left(\mathbf{0},\Sigma_{\Omega} \right).
\end{align*}
\end{theorem}
The proof is detailed out in Appendix \ref{prf: normality of omega}.

\begin{proposition}
\label{prop: asymp indep}
Define the statistic
\begin{align*}
\widehat \beta_{-E}^\perp=\frac1{ n} X_{-E}\tran(\nabla A(X_E\widehat\beta_E ) - Y).
\end{align*}
Then, 
\begin{align*}
\sqrt n
\begin{pmatrix}
\widehat\beta_{E}-\beta_{E,n} \\  \widehat\beta_{-E}^\perp \\ \Omega
\end{pmatrix}
\convdis \N_{p(K+1)}\left(\mathbf{0},\begin{pmatrix} 
\calI_{E,E}^{-1} &\mathbf{0} &\mathbf{0} \\
\mathbf{0} & \calI/\calI_{E,E} & \mathbf{0}\\
\mathbf{0} &\mathbf{0} & \Sigma_{\Omega}
\end{pmatrix}\right),
\end{align*}
where $\calI/\calI_{E,E}=\calI_{-E,-E}-\calI_{-E,E}\calI_{E,E}^{-1}\calI_{E,-E} $ is the Schur complement of $\calI_{E,E}$.
\end{proposition}
The proof is provided in Appendix \ref{prf: asymp indep}.

Our next result, Theorem \ref{sel:event:rep}, makes two contributions.
First, we identify a straightforward representation for (a subset of) our selection event.
We proceed by conditioning on this event to form the selective likelihood.
But we must characterize the unconditional distribution of variables in this representation before conditioning. Thus the second result is the asymptotic unconditional likelihood for a fixed set $E$ and fixed signs $S$, using the limiting law of variables in Theorem \ref{thm: normality of omega} and Proposition \ref{prop: asymp indep}.
Hereafter, we let $\Pi= \mathbb{E}[\ \widehat\Pi \ ]$, and similarly, $\kappa,\Theta,\Psi,\tau,\Gamma$ are the population versions of the matrices $\widehat\kappa,\widehat\Theta,\widehat\Psi,\widehat\tau,\widehat\Gamma$ that we defined in Section~\ref{sec:3}. 
We need the following assumption.
\begin{assumption}
\label{assump: densities}
Let the distribution of
$$\sqrt{n}\begin{pmatrix}
\widehat\beta_{E}-\beta_{E,n} \\  \widehat\beta_{-E}^\perp \\ \Omega
\end{pmatrix}$$
be absolutely continuous distributions on $\mathbb{R}^{p(K+1)}$, and let $p_n$ be the corresponding sequence of densities.
Assume that $p_n$ are absolutely equicontinuous and bounded.
\end{assumption}

The condition in Assumption \ref{assump: densities} together with the weak convergence in Proposition \ref{prop: asymp indep} implies that the densities $p_n$ converge to the corresponding limiting Gaussian density, uniformly on compact subsets of  $\mathbb{R}^{p(K+1)}$.

\begin{theorem}
\label{sel:event:rep}
The following assertions hold.
\begin{enumerate}[label=(\roman*).]
  \item $\left\{\widehat{E}^{(k)}=E^{(k)}, \  \widehat{S}^{(k)}=S^{(k)}, \ \widehat{Z}^{(k)} = Z^{(k)} \text{ for } k\in [K]\right\} = \left\{ \text{\normalfont sign}(\sqrt{n}\widehat{B})= S, \ \ \widehat{Z} = Z\right\}$.
  \item Under Assumption \ref{assump: densities}, the asymptotic distribution of
\begin{align}
\sqrt{n}\begin{pmatrix} \widehat\beta_{E} & \widehat \beta_{-E}^\perp & \widehat{B} \end{pmatrix} \Big\lvert \widehat{Z}=Z
\label{equ: betas cond Z}
\end{align}
leads us to the following asymptotic likelihood function
\begin{align*}
p_{Z}\left(\beta_{E,n}; \widehat\beta_{E}, \widehat\beta_{-E}^\perp,  \widehat{B}\right) \propto \varphi(\sqrt n\widehat\beta_E; \Pi \sqrt n\beta_{E,n} + \kappa; \Theta)\cdot \varphi(\sqrt n\widehat B; \Psi \sqrt n\widehat\beta_E + \tau;\Gamma).
\end{align*}
\end{enumerate}
\end{theorem}
The proof is provided in Appendix \ref{prf: sel:event:rep}. We emphasize that the asymptotic likelihood function on the right-hand-side does not depend on $\widehat\beta^\perp_{-E}$.

\subsection{Asymptotic selective likelihood}
\label{sec: asymp sel likelihood}
Now, we turn to obtain the conditional distribution of the aggregated MLE.
We begin with the the distribution of \eqref{equ: betas cond Z}
which gives us the likelihood function $p_{n, Z}\left(\beta_{E,n}; \widehat\beta_{E}, \widehat\beta_{-E}^\perp,   \widehat{B}\right)$.
Because of the representation of selection in Theorem \ref{sel:event:rep}, inference conditional on this event proceeds by truncating $p_{n, Z}\left(\beta_{E,n}; \widehat\beta_{E}, \widehat\beta_{-E}^\perp,   \widehat{B}\right)$
to the event $\{\sqrt{n}\widehat{B}\in \mathcal{O}_S\}$.
Formally, we can base conditional inference on
\begin{equation}
\begin{aligned}
& \dfrac{p_{n, Z}\left(\beta_{E,n}; \widehat\beta_{E}, \widehat\beta_{-E}^\perp,   \widehat{B}\right) }{\int p_{n, Z}\left(\beta_{n,E}; b_{E}, b_{-E}^\perp,  B\right) 1_{\mathcal{O}_S}(\sqrt{n}B)db_{E}  db_{-E}^\perp  dB} 1_{\mathcal{O}_S}(\sqrt{n}\widehat{B})\\
&=\dfrac{p_{n, Z}\left(\beta_{E,n}; \widehat\beta_{E}, \widehat\beta_{-E}^\perp,   \widehat{B}\right) }{\mathbb{P}\left[\sqrt{n}\widehat{B}\in \mathcal{O}_S  \;  \Big\lvert \; \widehat{Z}=Z\right]} 1_{\mathcal{O}_S}(\sqrt{n}\widehat{B}).
\label{likelihood:n}
\end{aligned}
\end{equation}
The log-likelihood based on \eqref{likelihood:n} is equal to 
$$
\log p_{n, Z}\left(\beta_{E,n}; \widehat\beta_{E}, \widehat\beta_{-E}^\perp,   \widehat{B}\right) -\log \mathbb{P}\left[\sqrt{n}\widehat{B}\in \mathcal{O}_S  \;  \Big\lvert \; \widehat{Z}=Z\right].
$$
Note, the first term in the log-likelihood can be replaced with its asymptotic counterpart $p_{Z}\left(\beta_{E,n}; \widehat\beta_{E}, \widehat\beta_{-E}^\perp,  \widehat{B}\right)$, which was derived in Theorem \ref{sel:event:rep}. 

The main result in this section is to approximate the second term by a large-deviation limit for the log-probability, under some moment and regularity conditions that are usually made to ensure existence of the limit.
Consider a real-valued sequence $a_n$ that goes to infinity as $n\to \infty$, and $a_n = o(n^{1/2})$.
Assume that $\sqrt{n}\beta_{E,n} = a_n\beta_E\in {\mathbb{R}^{|E|}}$, where $\beta_E$ does not depend on $n$.

\begin{assumption}[Moment condition and convergence of remainder]
\label{as:1}
Based on our proof for Proposition \ref{prop: asymp indep}, we have
\begin{align}
\sqrt n
\begin{pmatrix}
\widehat\beta_{E}-\beta_{E,n} \\  \widehat\beta_{-E}^\perp \\ \Omega
\end{pmatrix} =
\sqrt n \bar E_n + R_n,
\label{linear:rep}
\end{align}
where $\bar E_n= \frac{1}{n}\sum_{i=1}^n e_{i,n}$  is the average of $n$ i.i.d. observations, and $R_n=o_p(1)$.
Assume that 
$$\mathbb{E}\left[\exp(\lambda\|e_{1,n}\|_2)\right]<\infty$$
for some $\lambda\in \mathbb{R}^{+}$,
and that
\begin{equation*} 
\displaystyle\lim_{n\to \infty}  \frac{1}{a_n^{2}}\log \mathbb{P}\left[\frac{1}{a_n}\| R_n \|_2 > \epsilon \right] =- \infty 
 \end{equation*}
for any $\epsilon >0$, where $e_{i,n}$ and $R_n$ are based on the linear representation in \eqref{linear:rep}.
\end{assumption}

 \begin{assumption}
 \label{as:2}
 Consider the asymptotically linear representation in \eqref{linear:rep}.
For a fixed convex set $\mathcal{R}_0\subseteq \mathbb{R}^{p(K+1)}$, and for $O=O_p(1)$, we impose the condition that
\begin{equation*}
\begin{aligned} 
& \displaystyle\lim_{n\to \infty} \dfrac{1}{a_n^2} \left\{ \log \mathbb{P}\left[ \frac{1}{a_n}\begin{pmatrix}\sqrt{n}\widehat{\beta}_E \\ \sqrt{n}\widehat{\beta}^\perp_{-E} \\ \sqrt{n}\Omega \end{pmatrix} \in \mathcal{R}_0\right] - \log \mathbb{P}\left[ \frac{1}{a_n} \begin{pmatrix}\sqrt{n}\widehat{\beta}_E \\ \sqrt{n}\widehat{\beta}^\perp_{-E} \\ \sqrt{n}\Omega \end{pmatrix} + \frac{1}{a_n}O\in \mathcal{R}_0 \right]\right\}=0.
\end{aligned}
\end{equation*}
\end{assumption}

\begin{theorem}
Suppose that the conditions in Assumption \ref{as:1} and Assumption \ref{as:2} are met.
Define 
\begin{equation*}
\begin{aligned}
& L_n = \inf_{b, B} \Bigg\{\frac{1}{2}\left(b- \Pi \beta_E - \frac{1}{a_n}\kappa\right)\tran \Theta^{-1}\left(b- \Pi \beta_E - \frac{1}{a_n}\kappa\right)   \\
&\;\;\;\;\;\;\;\;\;\;\;\;\;\;\;+  \frac{1}{2}\left(B- \Psi b -\frac{1}{a_n}\tau\right)\tran \Gamma^{-1}\left(B- \Psi b -\frac{1}{a_n}\tau\right) + \frac{1}{a_n^2}\text{\normalfont Barr}_{\mathcal{O}_S}(a_n B)\Bigg\}.
\label{lim:convx}
\end{aligned}
\end{equation*}
Then, we have
\begin{align*}
& \displaystyle\lim_{n\to \infty} \dfrac{1}{a_n^2} \log \mathbb{P}\left[ \sqrt{n} \widehat{B}\in \mathcal{O}_S \;  \Big\lvert \; \widehat{Z}= Z\right] + L_n = C,
\end{align*}
where $C$ is a constant that does not depend on $\beta_E$.
\label{thm:consistency}
\end{theorem}

As a consequence of Theorem \ref{thm:consistency}, we can substitute the log-probability in the denominator of \eqref{likelihood:n} by 
\begin{equation}
\begin{aligned}
&-\inf_{b, B} \Bigg\{\frac{1}{2}\left(a_n b- a_n\Pi \beta_{E} - \kappa\right)\tran \Theta^{-1}\left(a_n b-a_n \Pi \beta_{E} - \kappa\right) \\
&\;\;\;\;\;\;\;\;\;\;\;\;+  \frac{1}{2}\left(a_n B- a_n\Psi  b -\tau\right)\tran \Gamma^{-1}\left(a_nB-a_n \Psi  b -\tau\right) + \text{Barr}_{\mathcal{O}}(a_n B)\Bigg\},
\label{opt:a}
\end{aligned}
\end{equation}
after we ignore the additive constant in the limit.
Finally, the optimization in the above display can be written as 
\begin{equation*}
\begin{aligned}
p_n(\beta_{E,n}) &=-\inf_{v, V} \Bigg\{\frac{1}{2}\left(\sqrt{n} v- \sqrt{n}\Pi \beta_{E,n} - \kappa\right)\tran \Theta^{-1}\left(\sqrt{n} v- \sqrt{n}\Pi \beta_{E,n} - \kappa\right) \\
&\;\;\;\;\;\;\;\;\;\;\;+  \frac{1}{2}\left(\sqrt{n} V- \sqrt{n}\Psi  v -\tau\right)\tran \Gamma^{-1}\left(\sqrt{n}V- \sqrt{n}\Psi v -\tau\right) + \text{Barr}_{\mathcal{O}}(\sqrt{n} V)\Bigg\},
\end{aligned}
\end{equation*}
by reparameterizing $a_n b$ and $a_n B$ in \eqref{opt:a} as $\sqrt{n}v$ and $ \sqrt{n}V$, respectively.
This yields us an asymptotic selective log-likelihood
\begin{equation}
\begin{aligned}
& \log\varphi(\sqrt{n}\widehat\beta_E;\Pi\sqrt{n}\beta_{E, n}+\kappa,\Theta) - p_n(\beta_{E,n}).
\label{asym:lik}
\end{aligned}
\end{equation}

The score and curvature of the asymptotic selective likelihood in \eqref{asym:lik} give us the selective MLE and the selective obs-FI matrix in Section \ref{sec:3}.
Note, the derivation of the two estimators follows the steps in \cite{panigrahi2022approximate} for the standard Gaussian regression problem.
We provide this result below in the interest of completeness.
\begin{theorem}
Consider solving the optimization problem
\begin{equation}
\label{optimization:inference:popn}
\widehat{V}_{\widehat\beta_E}^{\star}= \underset{V\in\R^{\bar d}}{\argmin}\; \frac{1}{2}(\sqrt{n}V - \Psi \sqrt{n}\widehat\beta_E -\tau)\tran \Gamma^{-1}(\sqrt{n}V - \Psi \sqrt{n}\widehat\beta_E -\tau) + \text{\normalfont Barr}_{\mathcal{O}_S}(\sqrt{n}V).
\end{equation}
The maximizer of the approximate selective likelihood and the observed Fisher information matrix are equal to
$$
 \Pi^{-1}\widehat\beta_E - \frac{1}{\sqrt{n}}\Pi^{-1}\kappa + \widehat\calI_{E,E}^{-1}\Psi\tran \Theta^{-1}\left(\Psi \widehat\beta_E +\frac{1}{\sqrt{n}}\tau -\widehat{V}_{\widehat\beta_E}^{\star}\right),
$$
$$
 \widehat\calI_{E,E}\left(\Theta^{-1}+ \Psi\tran \Gamma^{-1}\Psi - \Psi\tran \Gamma^{-1} \left(\Gamma^{-1}  + \nabla^2\text{\normalfont Barr}_{\mathcal{O}_S}\left(\sqrt{n}\widehat{V}_{\widehat\beta_E}^{\star}\right)\right)^{-1}\Gamma^{-1} \Psi\right)^{-1}  \widehat\calI_{E,E},
$$
respectively
\label{est:eqns}.
\end{theorem}

In practice, we use the empirical estimates of the matrices $\Pi,\kappa,\Theta,\Psi,\tau,\Gamma$ which yields us the expressions for the selective MLE and the selective obs-FI in \eqref{selective:MLE} and \eqref{obs:FI}, respectively.

\section{Addressing p-value lotteries}
\label{sec:5}

In sparse regression, construction of p-values is more feasible after the number of variables is reduced to a manageable size.
Sample-splitting, e.g., \cite{wasserman2009high}, is often a simple way to first select variables on a subsample of the full data, and then report the corresponding p-values, using classical least squares estimation on the remaining samples.
Variables that do not appear in the selected set are assigned a p-value equal to $1$.
More recently, a more powerful alternative to sample-splitting has been introduced in \cite{fithian2014optimal} via conditioning, which is called carving.
However, results produced by a single round of sample splitting or carving are overly sensitive, which comes from how one splits the data into two subsamples, leading to widely different  p-values. 
This problem has been reported in literature as the p-value lottery problem.
See, for example, the paper by \cite{meinshausen2009p}.

To address the p-value lottery problem, \citet{dezeure2015high} aggregate p-values from  repeated splitting, and more recently, \citet{schultheiss2021multicarving} propose to aggregate p-values after repeated carving on random splits of data. 
We refer to the former procedure as multi-splitting, and the latter procedure as multi-carving.
With increasing numbers of replicates, the results are expected to be less sensitive to the randomness introduced by the splits. 
Suppose one conducts multi-splitting or multi-carving $B$ times, and obtains the p-values $p_j^{(b)}$, for $b\in[B]$, and $j \in [p]$. 
This is followed by aggregating the $B$ p-values through their empirical quantiles
\begin{align*}
Q_j(\gamma):=q_{\gamma}\left(\left\{\frac{1}{\gamma}p_j^{(b)} ,\,b\in[B] \right\} \right)\wedge 1 ,
\end{align*}
where $q_\gamma$ denotes the $\gamma$-th empirical quantile. 
One can also minimize over $\gamma$ and use
\begin{align}
P_j:=\left[(1-\log(\gamma_{\min})) \inf_{\gamma\in(\gamma_{\min},1)} Q_j(\gamma)\right] \wedge 1.
\label{equ: gamma min}
\end{align}
The aggregation scheme produces valid p-values as long as p-values in each replicate are valid.

Below, we show that our procedure can be easily adapted to address the p-value lottery problem without recourse to MCMC sampling.
We proceed as multi-splitting and multi-carving, i.e., we use a subsample of size $n_1$ for variable selection.
For inference, we re-use data from selection by conditioning on the event of selection. 
We repeat this procedure $B$ times and aggregate the p-values as above. 
Moreover, in each replicate, we can  draw $K$ random subsets of size $n_1$ with replacement. 
A base model is selected using each subset, and the $K$ base models are aggregated as done in~\eqref{aggregate}. 
To computed p-values for the variables in our selected GLM, we can apply the same procedure as described in Section~\ref{sec:3}, with the matrices
\begin{align*}
&\{\widehat\Gamma^{-1}\}_{j,k}=\delta_{j,k}\frac{\rho}{1-\rho}\widehat\calI_{E^{(j)},E^{(j)}};\quad\{\widehat\Gamma^{-1}\widehat\tau \}_k=-\frac{\rho}{1-\rho} g_k^{(k)};\\
&\widehat\Theta^{-1}=\left(1+\frac{K\rho}{1-\rho}\right) \widehat\calI_{E,E}-\widehat\Psi\tran\widehat\Gamma^{-1}\widehat\Psi.
\end{align*}
Other matrices take the form that we provided in Section~\ref{sec:3}, with $\rho_k=\rho=\frac{n_1}{n}$ and $\rho_0=1-\rho_1$. 
Specially now, the distribution of the randomization variable in \eqref{randomization:distributed} is slightly different, which results in different expressions for $\widehat\Gamma$, $\widehat\tau$, and $\widehat\Theta$.
A derivation of the asymptotic distribution for randomization is deferred to Lemma \ref{prop: w r} in Appendix~\ref{sec: with replacement}.

\section{Experiments}
\label{sec: experiments}

This section provides numerical justifications for our proposed procedure.

\subsection{Experiments with distributed datasets}

We simulate our data according to two main models.
For $i\in [n]$, we draw $x_i\sim\N_p(0,\Sigma)$, where $\Sigma_{ij}=\rho^{|i-j|}$ with $\rho=0.9$, $p=100$. 
In our first model, we draw a real-valued response from a linear model as  $y_i\sim\N(x_i\tran\beta,1)$.
In our second model, we draw a binary response from a logistic-linear model as 
$y_i\sim \text{Bernoulli}(1/(1+e^{-x_i\tran\beta}))$.
Observation $i$ is independent of all the other observations in our dataset.
There are 5 non-zero coefficients in our model; each non-zero $\beta_j$ is equal to $\pm\sqrt{2c\log p}$, where the sign is randomly determined in both models.
In the remaining section, we call parameter $c$ as the ``signal strength".

In every round of simulation, we partition the full dataset into $K+1$ disjoint subsets
$$\left\{D^{(k)} \text{ for } k\in \{0\}\cup [K]\right\}.$$ 
Subsets 1 through $K$ are used for variable selection, and subset 0 is used only at the time of selective inference. 
Equivalently, in our setup, $D^{(k)}$ is allocated to a local machine, for $k\in [K]$, and $D^{(0)}$ is accessed only by the central machine for selective inference.
We use an extra dataset to tune the regularization parameter $\lambda$ for model selection by sweeping over a grid of values $\left\{t \sqrt{2\log p} \cdot \text{sd}(Y)\mid t=0.5,1,\ldots,5\right\}$; we do not use this dataset further, either to select predictors through a generalized linear regression, or to infer for the selected predictors.
In the first model, we run the usual linear regression with the quadratic loss function, and infer in the selected linear model.
In the second model, we run a logistic regression on each local machine, and base inference on the selected logistic-linear model.

We design three different scenarios to investigate the performance of our procedure over $500$ rounds of simulations.
\begin{enumerate}[label=(\Roman*).]
\item In Scenario 1, we vary the number of distributed datasets $K\in\{2,4,6,8\}$. Each local machine uses $n_k=[8000/K]$ samples, and the central machine has access to $1000$ samples. We fix the signal strength at $c=0.1$.

\item In Scenario 2, we consider $3$ distributed datasets. 
The central machine has $2000$ samples, and each of the two local machines has $4000$ samples.
We investigate four signal regimes by varying $c\in \{0.3,0.5,0.7,0.9\}$; we number these regimes as $1-4$.

\item  In Scenario 3, we vary the number of samples that are reserved only for selective inference at the central machine; this number takes a value in the set \sloppy{$\{250,500,1000,2000\}$}. Each of the three local machines has $2000$ samples, and the signal strength is fixed at $c=0.5$ in this setting.
\end{enumerate}

Figures \ref{fig: linear} and \ref{fig: logistic} summarize the results of our simulations, in a linear and logistic regression problem.
We begin by evaluating the coverage of $90\%$ confidence intervals that are centered around the selective MLE, and with variance estimated by the entries of the selective obs-FI matrix. 
We call our method ``Dist-SI" in the plots. 
As a baseline for comparison, we consider ``Splitting" which means that we simply use the samples at the central machine to infer in the GLM selected by the local machines, using the standard Wald confidence intervals.
For the selected GLM in \eqref{sel:model}, we note that our parameter of interest is 
$$\beta^{E}=\underset{b}{\argmin} \; \EE{A(x_{i,E}\tran b)-y_ix_{i,E}\tran b},$$ 
the minimizer of the generalized linear regression problem with the selected predictors. 

A comparison of power is provided next.
We compute the lengths of confidence intervals, which indicate the power associated with selective inference for $\beta^E$.
We follow this up by calculating the power of correctly detecting a true signal in $\beta$. 
A signal is detected if it is selected in our model, and the selective confidence intervals do not cover $0$.

\textbf{Observations}. Across all scenarios, the confidence intervals produced by ``Dist-SI" (approximately) attain the desired coverage probability. 
``Splitting" produces valid confidence intervals, but discards samples used by the local machines. 
The advantages of re-using data from the local machines are quite evident in the plots for the lengths of the confidence intervals, and in the plots for the fraction of times that they detect a true signal.
As expected, ``Dist-SI" yields tighter confidence intervals and achieves higher power than the baseline procedure based on ``Splitting" in all three scenarios.
More specifically, we observe that interval lengths for both methods increase with $K$ in Scenario 1.
This is because the final model, which is the union of the $K$ models selected by local machines, is likely to be larger for larger $K$. In this case, the variance of $\widehat\beta_j$ tends to be larger.
Consistent with standard expectations, the power, for both methods, has an increasing trend as signal strength $c$ increases in Scenario 2.
In Scenario 3, we see that both methods produce longer intervals when $n_0$ decreases, and as expected, the gap between the baseline and ``Dist-SI" is more pronounced with fewer samples at the central machine.

\begin{figure}
\centering
\begin{subfigure}{.8\textwidth}
\centering
\includegraphics[width=\textwidth]{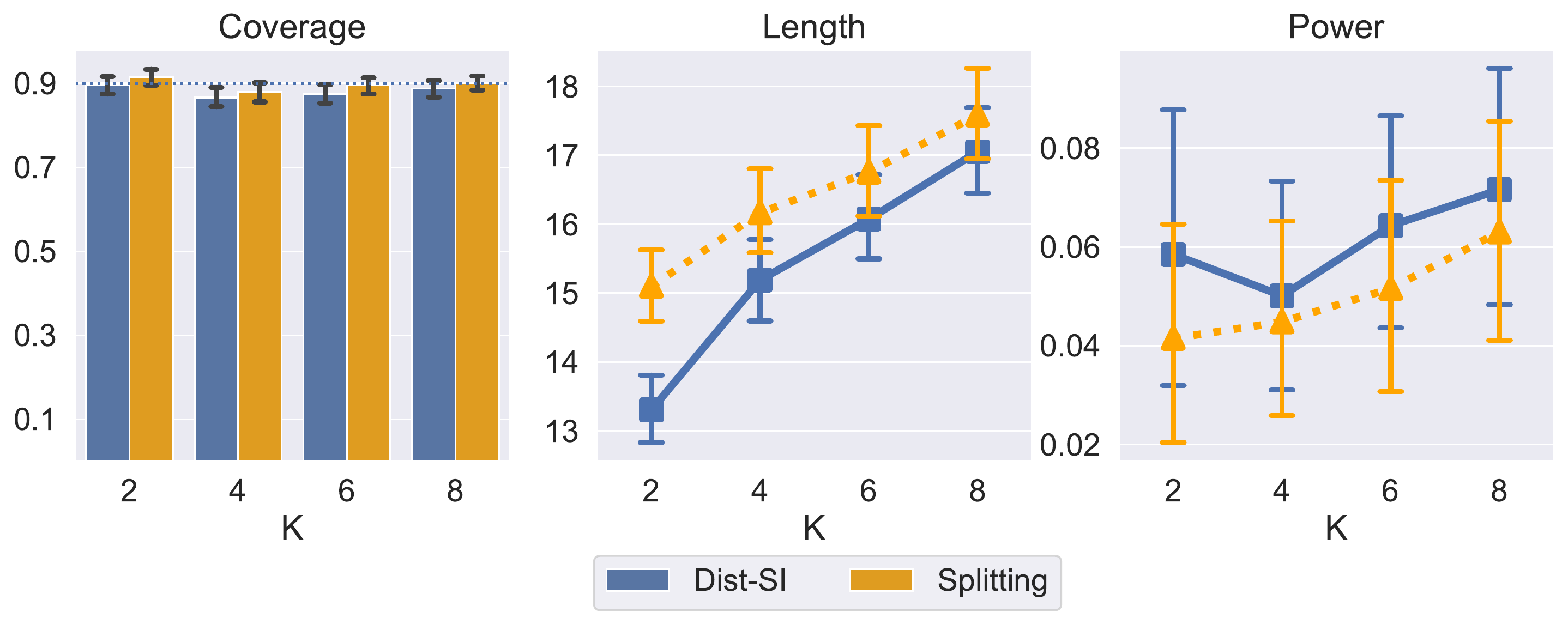}
\caption{Varying $K$. Each local machine has $[8000/K]$ data points for each $K$. 
}
\label{fig: linear vary K}
\end{subfigure}
\begin{subfigure}{.8\textwidth}
\centering
\includegraphics[width=\textwidth]{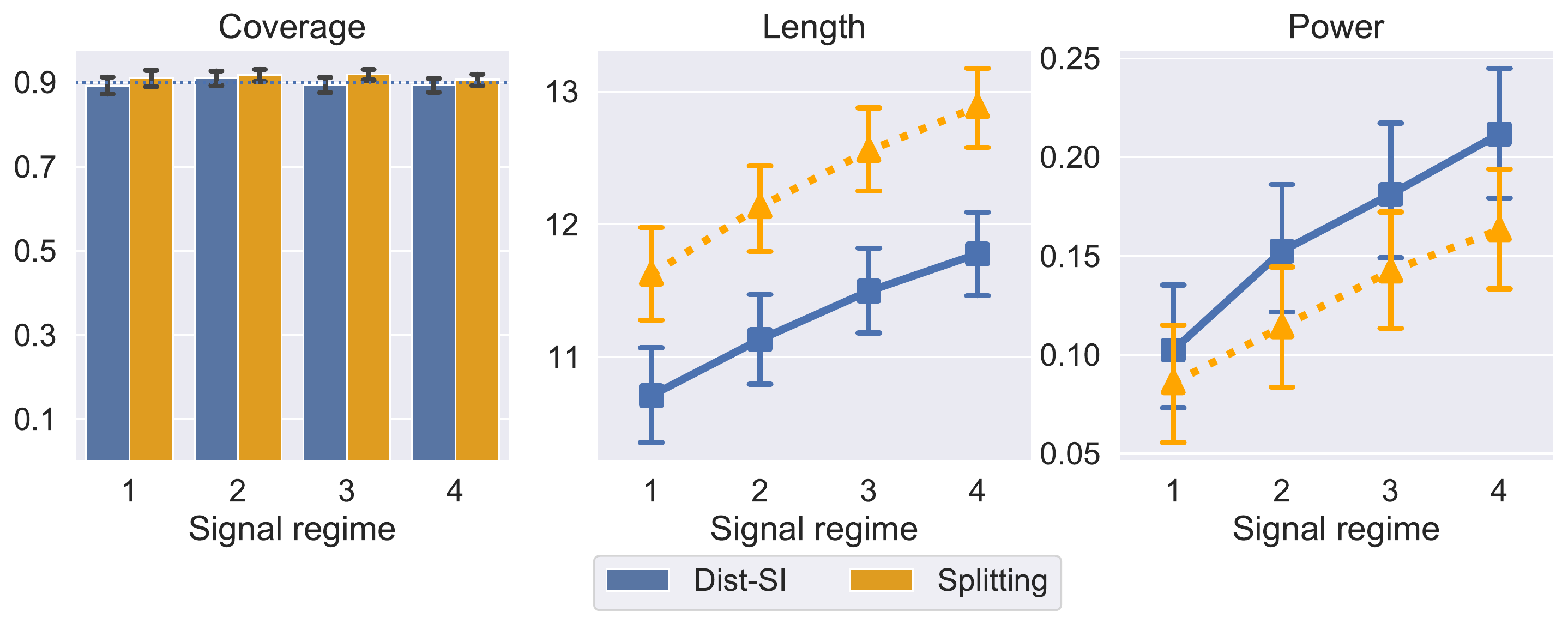}
\caption{Varying signal strength. The nonzero $\beta_j$ equals $\pm\sqrt{2c\log p}$ with random signs for $c=0.3,0.5,0.7,0.9$ in the four signal regimes.}
\label{fig: linear vary signal}
\end{subfigure}
\begin{subfigure}{.8\textwidth}
\centering
\includegraphics[width=\textwidth]{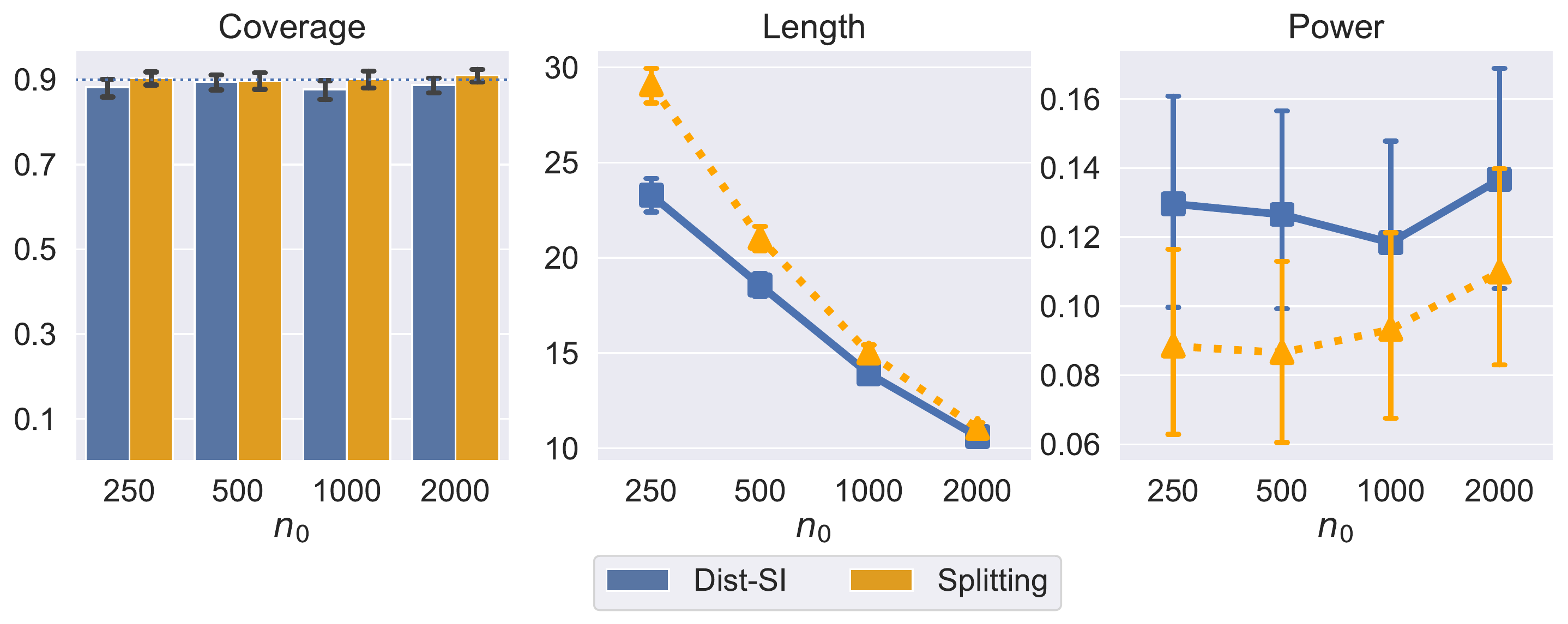}
\caption{Varying $n_0$, the sample size in the central machine.}
\label{fig: linear vary n0}
\end{subfigure}
\caption{Results for linear regression.}
\label{fig: linear}
\end{figure}

\begin{figure}
\centering
\begin{subfigure}{.8\textwidth}
\centering
\includegraphics[width=\textwidth]{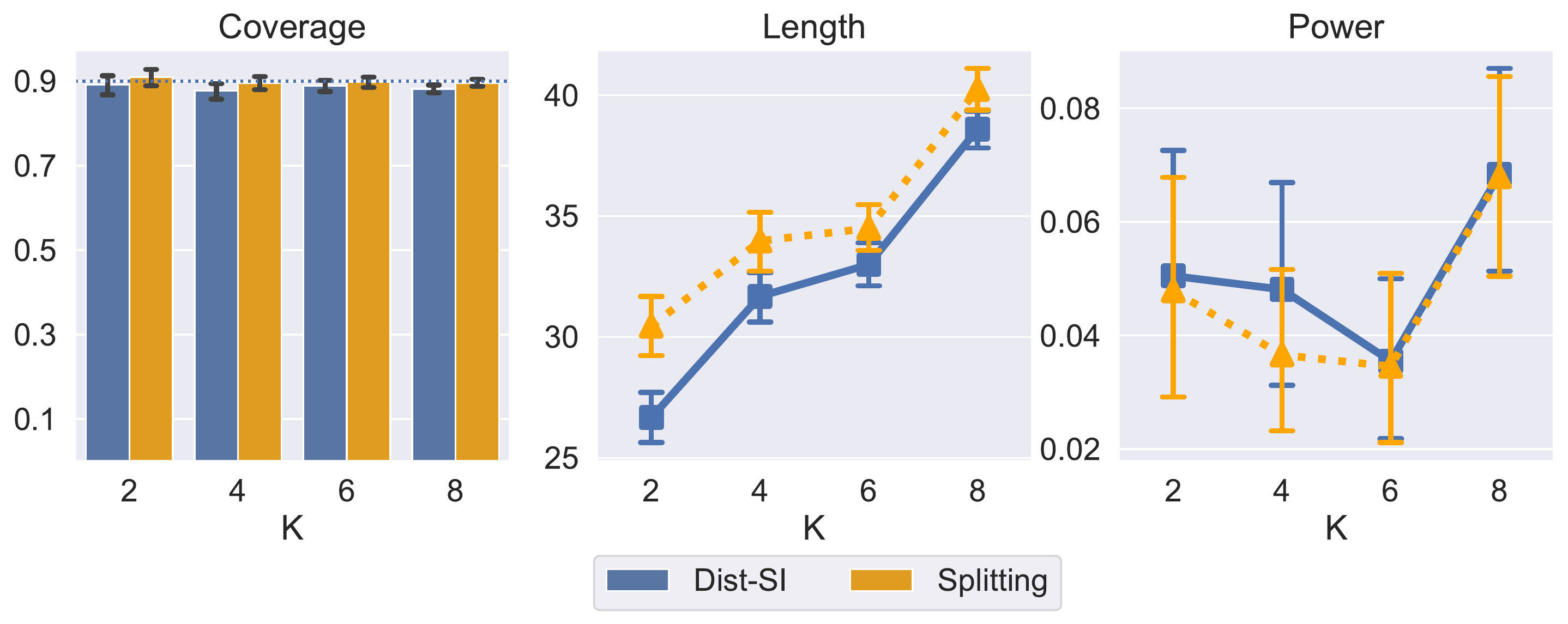}
\caption{Varying $K$. Each local machine has $[8000/K]$ data points for each $K$. }
\end{subfigure}
\begin{subfigure}{.8\textwidth}
\centering
\includegraphics[width=\textwidth]{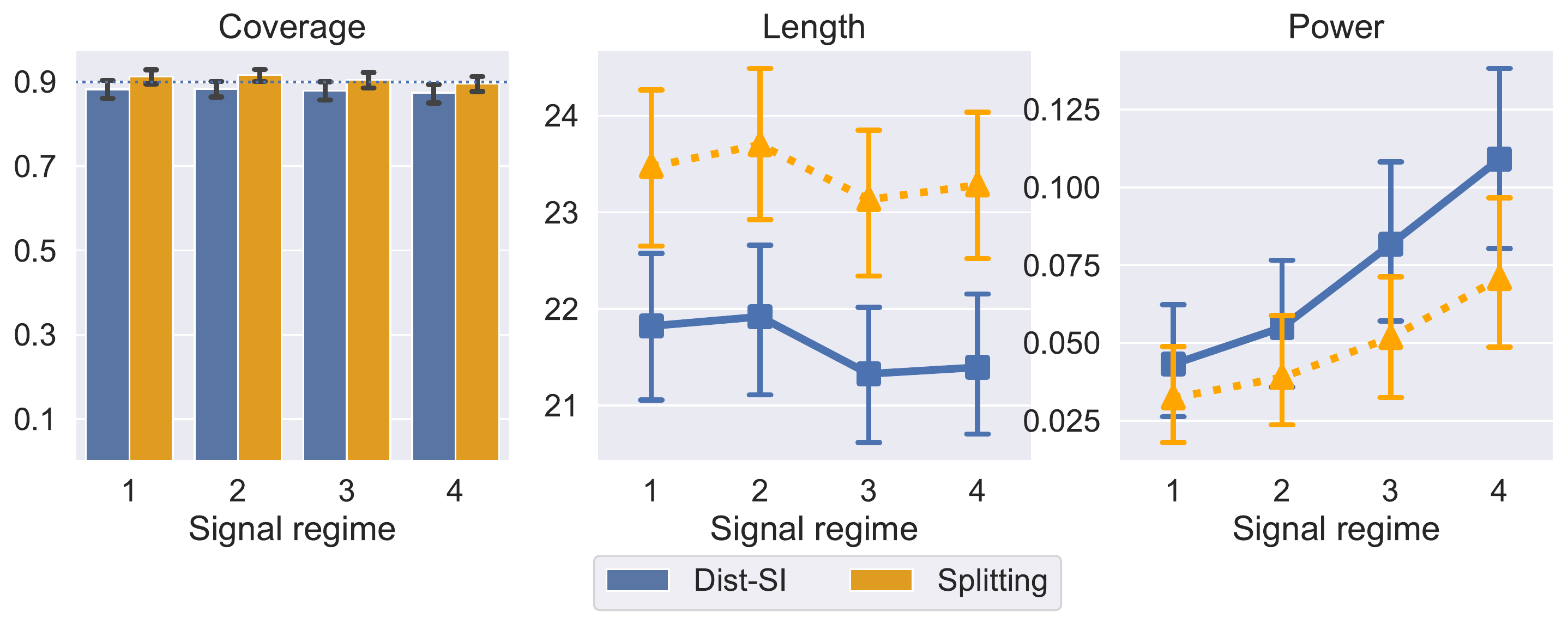}
\caption{Varying signal strength. The nonzero $\beta_j=\pm\sqrt{2c\log p}$ with random signs for $c=0.3,0.5,0.7,0.9$.}
\end{subfigure}
\begin{subfigure}{.8\textwidth}
\centering
\includegraphics[width=\textwidth]{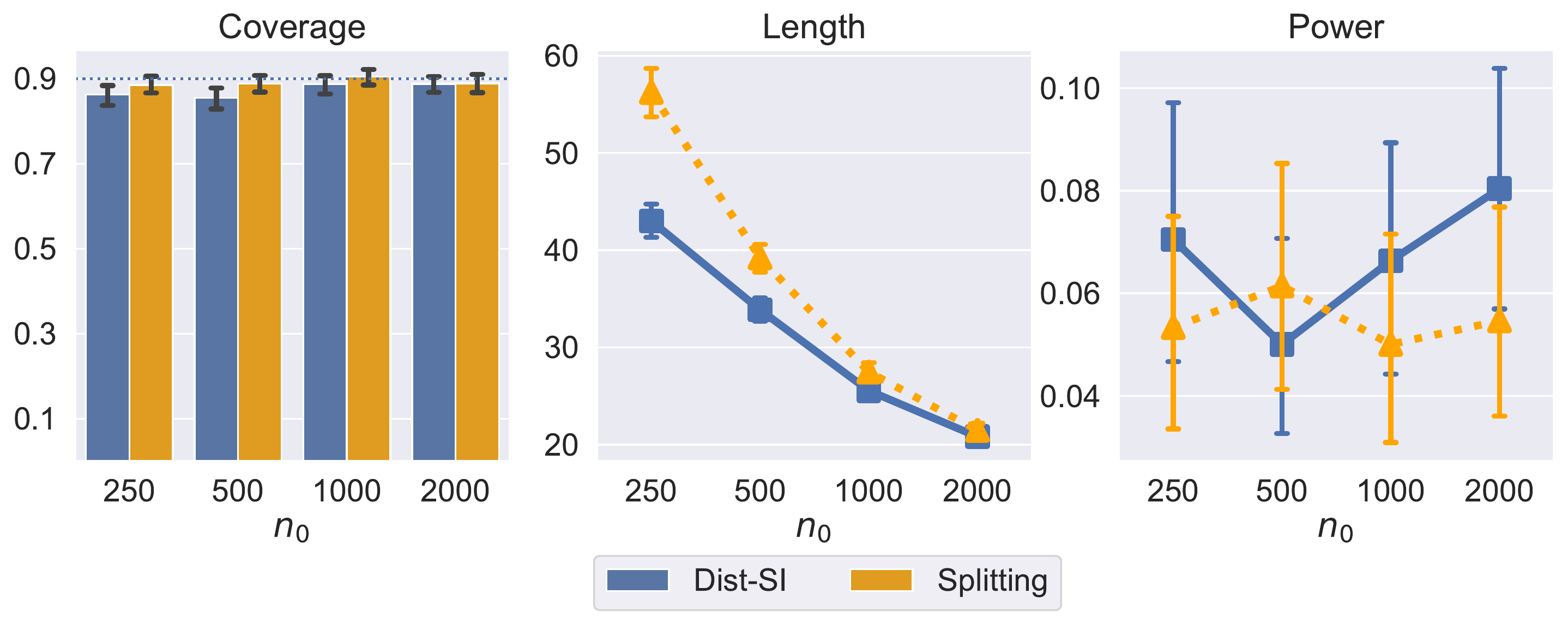}
\caption{Varying $n_0$, the sample size in the central machine.}
\end{subfigure}
\caption{Results for logistic regression.}
\label{fig: logistic}
\end{figure}

\subsection{Experiments on P-value lotteries}

In this section, we apply the suitable adaptation of our procedure to solve the p-value lottery problem, as described in Section \ref{sec:5}. 
We compare our procedure with ``Multi-carving" and ``Multi-splitting" as proposed by  \cite{schultheiss2021multicarving} and \cite{dezeure2015high}, respectively. 
For the latter two algorithms, we use the implementation provided by \cite{schultheiss2021multicarving} with code available on GitHub\footnote{\url{https://github.com/cschultheiss/Multicarving}.
The original code is written in R, and we load them into Python when running our simulations, which might have contributed to slightly longer running times as reported in our  findings.}.
To avoid any confusion, we continue to refer our procedure as ``Dist-SI", though we are no longer simulating distributed datasets.

We generate our data from the same linear model as described before, but now we use the sample size, dimension and sparsity regime that was discussed in \cite{schultheiss2021multicarving}.
That is, we fix $n=100$ and $p=200$, and consider $20$ nonzero coefficients with $\beta_j=\pm 2\sqrt{\log p}$.
We use $B=5$ replicates, and aggregate the p-values using formula~\eqref{equ: gamma min} with $\gamma_{\min}=0.05$. 
The proportion of samples used for variable selection is varied in the set $\{0.5,0.6,\ldots,0.9\}$.

We fix the significance level at $0.1$.
A coefficient $\beta_j$ is predicted to be nonzero if $P_j<0.1$.
To compare the quality of p-values, we measure their accuracy in terms of the diagnostic odds ratio (DOR), which is defined as:
\begin{align*}
\text{DOR}:=\frac{\text{True positive}\cdot \text{True negative}}{\text{False positive}\cdot \text{False negative}}.
\end{align*}
Besides computing the DOR, we compare the average run time for ``Dist-SI" and ``Multi-carving".
The results are shown in Figure~\ref{fig: multisplit}.
In the left panel, we plot the diagnostic odds ratio of the three methods with varying proportions. 
The error bars are once again reported for $500$ random repetitions. 
In the right panel, we plot the average log-run times of ``Dist-SI" and ``Multi-carving". 

\textbf{Observations}. 
We find that our procedure has larger DOR than the two previously proposed alternatives, ``Multi-carving", and  ``Multi-splitting", for all values of sample proportion.
Especially, a p-value in every replicate uses the full data after carefully discarding information that was used up for selecting predictors.
The re-use of data from selection results in larger power over ``Multi-splitting".
Our procedure aligns with ``Multi-carving", which also deploys conditional techniques to re-use data for hypothesis testing.
However, a key distinction of our procedure with ``Multi-carving" lies in how we use the randomization framework to represent selection, and subsequently marginalize over this randomness to construct our p-values. 
In particular, we note that ``Multi-carving" conditions on the randomization that is involved during variable selection on a random split of the data, whereas our procedure explicitly characterizes the distribution of randomization instead of simply conditioning on $\Omega$.
We believe that this difference between the two procedures shows up in our simulated findings as we note larger values of DOR with ``Dist-SI". 
Unsurprisingly, our proposal is also faster than ``Multi-carving" by about $100$ times.
From a computing perspective, our procedure solves a convex optimization problem to deliver p-values; the latter procedure uses MCMC sampling from a conditional distribution for the same problem. 

\begin{figure}
\centering
\includegraphics[width=.9\textwidth]{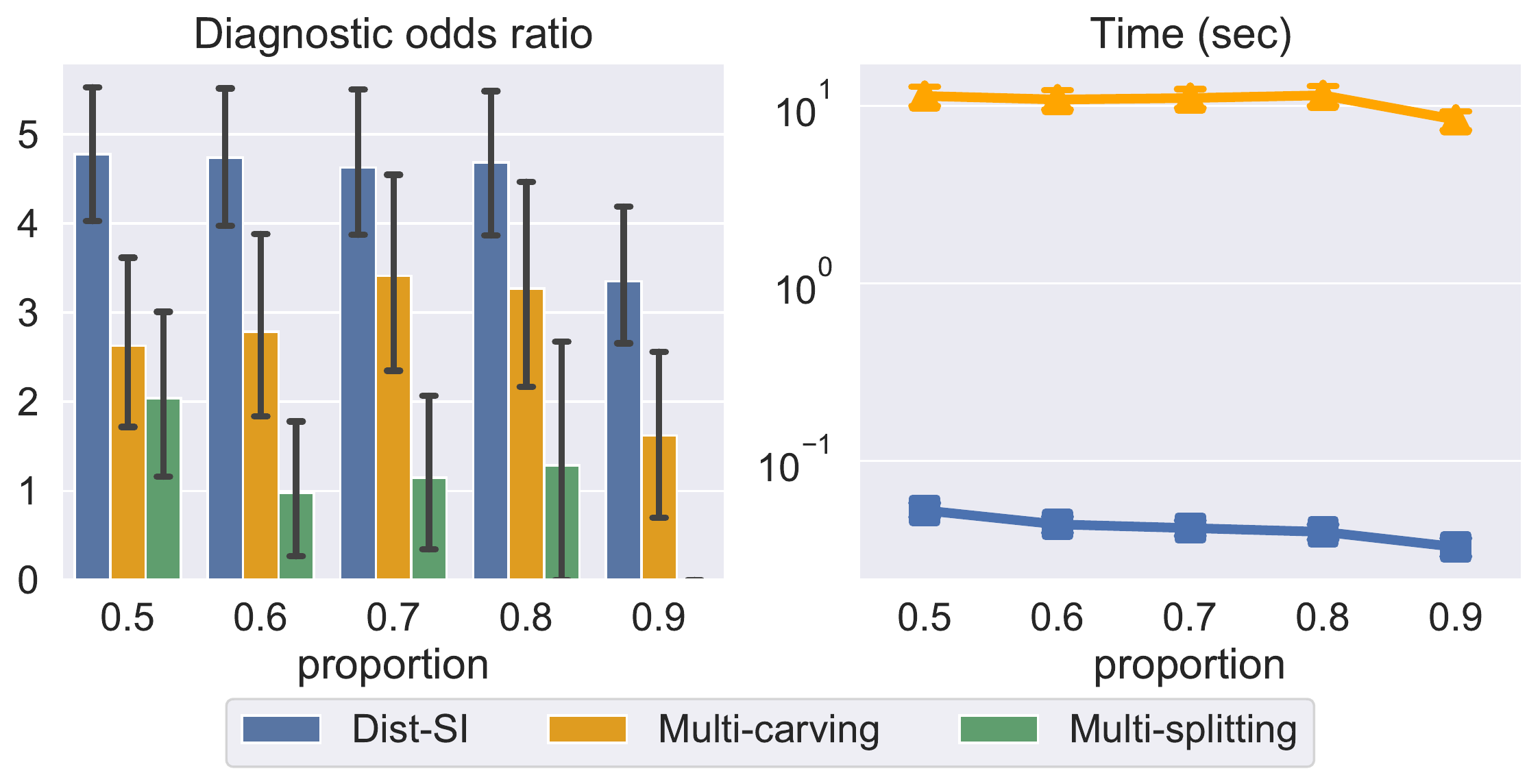}
\caption{Compare Dist-SI with multi-carving and multi-splitting. The left panel shows the DOR when using different proportions of samples for selection. The right panel shows the average running time.}
\label{fig: multisplit}
\end{figure}

\subsection{Experiments on medical dataset}

We illustrate an application of our procedure on a real dataset that is publicly available on MIT's GOSSIS database \cite{raffa2022global}.
This dataset contains records on intensive care unit (ICU) admissions from $192$ hospitals, including patients' demographic information, and various medical measurements, and lab results. 
We only use the datasets from the four largest hospitals, among which three datasets are used for variable and the remaining one is reserved for selective inference. 
We focus on a regression problem with data from the first 24 hours of intensive care.
The response in this problem is binary, and takes the value $1$ if a patient admitted to an ICU has been diagnosed with Diabetes Mellitus, and is $0$ otherwise.
The same problem appeared in the 2021 Women in Data Science Datathon \footnote{\url{https://www.kaggle.com/competitions/widsdatathon2021/data}. Accessed on on Dec. 17, 2022.}. 
We remove variables with more than half missing values, and also remove rows with missing values.
After preprocessing, we end up with 81 predictors. 
The three datasets used for variable selection have sample sizes ranging from $1633$ to $1788$, and the dataset reserved for inference has $2000$ samples.

For model selection, we run the logistic regression with Lasso penalty.
Consistent with our simulated experiments, the regularization parameter is tuned with one extra dataset with $893$ samples. 
The selected GLM has 58 predictors. 
To construct confidence intervals for the 58 selected variables, we apply the proposed ``Dist-SI" algorithm and ``Splitting" as done in simulations. 
The significance level is set to be $0.1$.
``Dist-SI" reports 21 significant variables, while ``Splitting" reports 13 significant variables. 
In Figure~\ref{fig: real data}, we plot the confidence intervals for the regression coefficients that are rejected by either of the two procedures. 
The boxplot for the lengths of these intervals, in Figure~\ref{fig: diabete lengths}, show that the median length of the ``Dist-SI" intervals is smaller than the ``Splitting" intervals by $67\%$. 
Additionally, the coefficient of variation is $1.8$ and $3.9$ for ``Dist-SI" and "Splitting", respectively.
This indicates that the dispersion of interval lengths for ``Dist-SI" is smaller  than ``Splitting".
On this instance, we see that ``Splitting" yields a few very wide intervals. 
This is because the Hessian matrix based on data present at the central machine (reserved dataset) is ill-conditioned. 
``Dist-SI" does not have this issue because it re-uses data from the three hospitals
for more powerful selective inference.
\begin{figure}
\centering
\includegraphics[width=\textwidth]{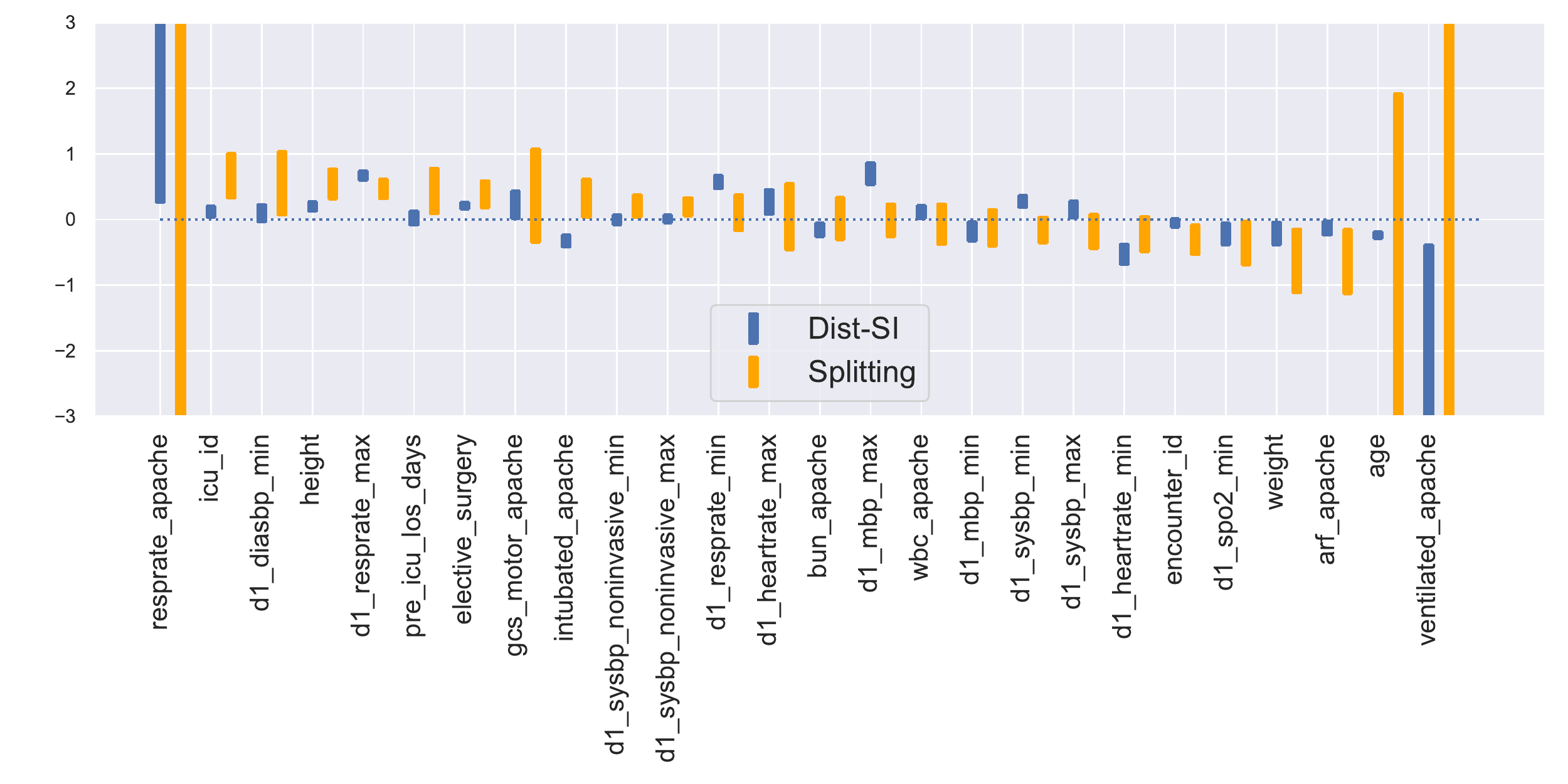}
\caption{Confidence intervals for the coefficients that are rejected by either Dist-SI or sample splitting.}
\label{fig: real data}
\end{figure}

\begin{figure}
\centering
\includegraphics[width=.5\textwidth]{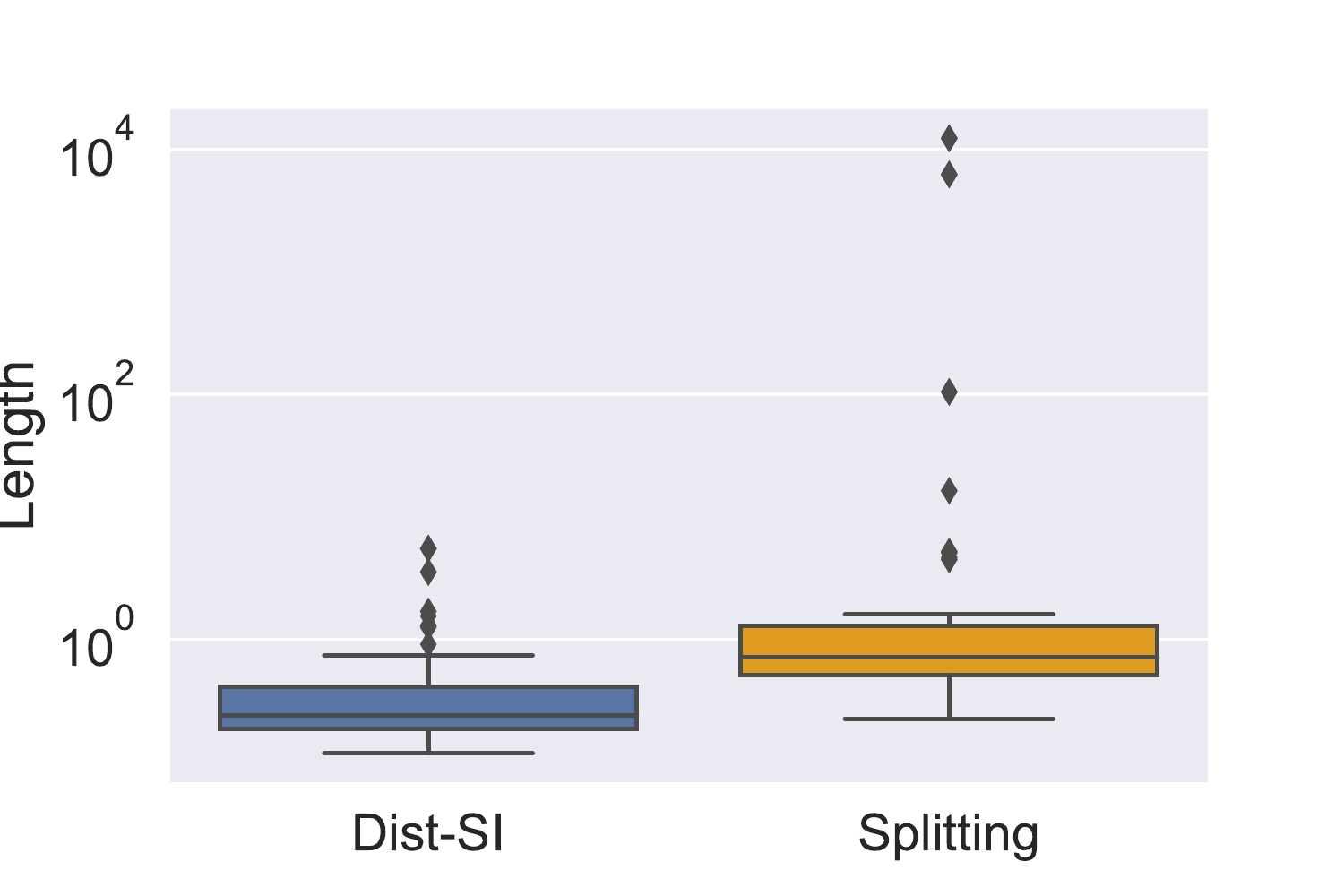}
\caption{Boxplot of confidence interval lengths produced by Dist-SI and sample splitting. The y-axis is on the logarithmic scale. }
\label{fig: diabete lengths}
\end{figure}

\section{Conclusion}
\label{sec: conclusion}

Model selection appears to be routine practice when analyzing big datasets. 
Inference for data-dependent models and parameters is a very challenging goal, because sound procedures must rigorously account for randomness from the selection process.
To the best of our knowledge, this is the first contribution that addresses selective inference with distributed data.
We provide a procedure to construct confidence intervals and p-values when inference is sought in a generalized linear model with selected predictors.
We identify a representation for selection in a common distributed setup, and provide an asymptotic selective likelihood by developing a randomized framework for our problem.
Approximately-valid selective inference, based on our selective likelihood, takes a very simple form: our confidence intervals for the selected regression coefficients are centered around the MLE of the selective likelihood, and the variance of the MLE is estimated by the observed Fisher information matrix.
An appealing feature of our procedure is that we only require some aggregated information, with relatively low communication cost, from each machine. 
This feature allows an adaptation of our procedure to settings where various data-sites may not be willing to share their individual datasets.
But, we note that there is room for improvement here, specially if various sites have not measured the same set of predictors.
Our paper also provides an efficient solution for the p-value lottery problem without relying on MCMC samplers.
Our procedure bypasses the primary computational bottleneck in the earlier proposal  \citep{schultheiss2021multicarving} by reducing selective inference to the solution of an optimization problem.

\section*{Acknowledgements}
S. Panigrahi's research is supported in part by NSF grants: NSF-DMS 1951980 and NSF-DMS 2113342. S. Liu's research is partially supported by the Stanford Data Science Scholars program.


\newpage
\appendix
\section*{SUPPLEMENTARY MATERIAL}

\section{Proofs for Section~\ref{sec: randomization}}

Supporting results are collected in Appendix \ref{appendix:supp}.
\subsection{Proof of Theorem~\ref{thm: normality of omega}}
\label{prf: normality of omega}

\begin{proof}
We start with the decomposition
\begin{align*}
\frac1n X\tran (\nabla A(X\widehat\beta^{\Lambda,(k)}) - Y)&=\frac1n X\tran(\nabla A(X_E \beta_{E,n}) - Y) + R_1^{(k)} + R_2^{(k)},
\end{align*}
where
\begin{align*}
R_1^{(k)}&=\frac1n X\tran (\nabla A(X_{E^{(k)}}\beta^*_{E^{(k)}, n} ) - \nabla A(X_E\beta_{E,n})),\\
R_2^{(k)}&=\frac1n X\tran (\nabla A(X\widehat\beta^{\Lambda,(k)}) - \nabla A(X_{E^{(k)}}\beta^*_{E^{(k)}, n} ) ),\text{ and } \\
\beta^{*}_{E^{(k)}, n}&=\calI_{E^{(k)},E^{(k)} }^{-1} \calI_{E^{(k)} ,E}\beta_{E,n}.
\end{align*}
In a similar fashion, we can decompose the variables based on $D^{(k)}$ as
\begin{align*}
\frac1{n_k}X^{(k),\intercal}(\nabla A(X^{(k)} \widehat\beta^{\Lambda,(k)} ) - Y^{(k)} )&=\frac1{n_k} X^{(k),\intercal}(\nabla A(X^{(k)}_E\beta_{E,n} ) - Y^{(k)} ) + r_1^{(k)} + r_2^{(k)}.
\end{align*}
The decomposition in the above two displays allow us to write
\begin{equation}
\begin{aligned}
\sqrt n\omega^{(k)}&=\sqrt n\Big\{\frac{1}{n} X\tran(\nabla A(X_E \beta_{E,n}) - Y) - \frac{1}{n_k}X^{(k),\intercal}(\nabla A(X^{(k)}_E\beta_{E,n}) -Y^{(k)} ) \\
&\;\;\;\;\;\;\;\;\;\;\;\;\;+ R_1^{(k)} + R_2^{(k)} - r_1^{(k)} - r_2^{(k)}\Big\}\\
&=\sqrt n\tilde\omega^{(k)} +\sqrt n R^{(k)},
\end{aligned}
\label{rep:omega}
\end{equation}
where 
$$R^{(k)}= R_1^{(k)} + R_2^{(k)} - r_1^{(k)} - r_2^{(k)},$$
and 
$$\tilde\omega^{(k)}= \frac{1}{n} X\tran(\nabla A(X_E \beta_{E,n}) - Y) - \frac{1}{n_k}X^{(k),\intercal}(\nabla A(X^{(k)}_E\beta_{E,n}) -Y^{(k)} ).$$

Let $\tilde \Omega\in\R^{pK}$ be the stack of $\tilde\omega^{(k)}$ for $1\leq k\leq K$.
It suffices to show that
\begin{align}
\sqrt n\tilde\Omega \convdis \N(\mathbf{0},\Sigma_{\Omega}), \quad\text{and}\quad \sqrt n R^{(k)} \convp0.
\label{equ: thm 4.1 sufficient}
\end{align}
To proceed with the proof, let $e_i=x_i(\nabla A(x_{i,E}\tran \beta_{E,n}) - y_i)$. 
It is easy to see that $e_i$ are i.i.d. for all $1\leq i\leq n$ with $\EE{e_i}=0$ and $\Var{e_i}\to\calI$, and it follows that
\begin{align*}
\sqrt n\tilde\omega^{(k)}&=\sqrt{1-\rho_k} \frac{1}{\sqrt{n-n_k}}\sum_{i\notin \calC_k} e_i - \frac{1-\rho_k}{\sqrt{\rho_k}} \frac{1}{\sqrt{n_k}}\sum_{j\in \calC_k} e_j.
\end{align*}
Clearly,
$$\sqrt n\tilde\omega^{(k)}\convdis \N\left(\mathbf{0}, \frac{1-\rho_k}{\rho_k}\calI\right),$$
and
\begin{align*}
\Cov{\sqrt n\tilde\omega^{(j)},\,\sqrt n\tilde\omega^{(k)} }&\to -\calI
\end{align*}
for $j\neq k$.
This leads us to claim
\begin{align*}
\sqrt n\begin{pmatrix}
\tilde\omega^{(j)} \\ \tilde\omega^{(k)}
\end{pmatrix}\convdis\N\left(\mathbf{0}, \begin{pmatrix}
\frac{1-\rho_j}{\rho_j} \calI & -\calI \\ -\calI & \frac{1-\rho_k}{\rho_k}\calI
\end{pmatrix} \right),
\end{align*}
which proves the first statement of \eqref{equ: thm 4.1 sufficient}.
Lemma~\ref{lem: R_1 - r_1} and Lemma~\ref{lem: R_2 - r_2} show that $\sqrt n R^{(k)}=o_p(1)$ to conclude the proof of \eqref{equ: thm 4.1 sufficient}.

\begin{lemma}[Rate of $R^{(k)}_1 - r^{(k)}_1 $]
\label{lem: R_1 - r_1}
Let
\begin{align*}
R_1^{(k)}&=\frac1n X\tran (\nabla A(X_{E^{(k)}}\beta^*_{E^{(k)}, n} ) - \nabla A(X_E\beta_{E,n})),\\
r_1^{(k)}&=\frac1{n_k} X^{(k),\intercal} (\nabla A(X^{(k)}_{E^{(k)}}\beta^*_{E^{(k)}, n} ) - \nabla A(X^{(k)}_E\beta_{E,n})).
\end{align*}
Then $\sqrt n(R^{(k)}_1 - r^{(k)}_1)=o_p(1)$.
\end{lemma}
\begin{proof}[Proof of Lemma~\ref{lem: R_1 - r_1}.]
As per Lemma~\ref{lem: property of beta*}, we have
$$X_{E^{(k)}}\beta_{E^{(k)}, n}^* - X_E\beta_{E,n}=O_p(n^{-1/2}).$$
Thus, we claim that
\begin{align*}
R_1^{(k)}-r_1^{(k)}
&=\frac1n X\tran \diag(\nabla^2 A(X_E\beta_{E,n}) (X_{E^{(k)}}\calI_{E^{(k)},E^{(k)}}^{-1}\calI_{E^{(k)},E } - X_{E} )\beta_{E,n}   \\
&-\frac1{n_k} X^{(k),\intercal} \diag(\nabla^2 A(X_E^{(k)}\beta_{E,n}) (X_{E^{(k)}}^{(k)}\calI_{E^{(k)},E^{(k)}}^{-1}\calI_{E^{(k)},E} - X_{E} )\beta_{E,n} + o_p(n^{-1/2}).
\end{align*}
By Assumption~\ref{assump: missed variables}, there exists $\calE_k\subseteq \widetilde{E}^{(k)}=E\setminus E^{(k)}$ such that for $j\in\widetilde{E}^{(k)}\setminus \calE_k $, $X_{E^{(k)}}\calI_{E^{(k)},E^{(k)}}^{-1}\calI_{E^{(k)},j}=X_j $ and $\beta_{\calE_k,n}=O(n^{-1/2})$.
So the last display simplifies as
\begin{align*}
R_1^{(k)}-r_1^{(k)}&=\frac1n X\tran \diag(\nabla^2 A(X_E\beta_{E,n})) (X_{E^{(k)}}\calI_{E^{(k)},E^{(k)}}^{-1}\calI_{E^{(k)},\calE_k } - X_{\calE_k } )\beta_{\calE_k,n }   \\
&-\frac1{n_k} X^{(k),\intercal} \diag(\nabla^2 A(X_E^{(k)}\beta_{E,n})) (X_{E^{(k)}}^{(k)}\calI_{E^{(k)},E^{(k)}}^{-1}\calI_{E^{(k)},\calE_k} - X_{\calE_k} )\beta_{\calE_k,n} + o_p(n^{-1/2}).
\end{align*}

If $\calE_k$ is not empty, let 
$$T_1=\EE{\frac1n X\tran \diag(\nabla^2 A(X_E\beta_{E,n}) (X_{E^{(k)}}\calI_{E^{(k)},E^{(k)}}^{-1}\calI_{E^{(k)},\calE_k} - X_{\calE_k})}.$$
Then 
\begin{align*}
R_1^{(k)}-r_1^{(k)}&=\left[\frac1n X\tran \diag(\nabla^2 A(X_E\beta_{E,n}) (X_{E^{(k)}}\calI_{E^{(k)},E^{(k)}}^{-1}\calI_{E^{(k)},\calE_k } - X_{\calE_k } ) - T_1\right] \beta_{\calE_k,n }   \\
&-\left[\frac1{n_k} X^{(k),\intercal} \diag(\nabla^2 A(X_E^{(k)}\beta_{E,n}) (X_{E^{(k)}}^{(k)}\calI_{E^{(k)},E^{(k)}}^{-1}\calI_{E^{(k)},\calE_k} - X_{\calE_k} ) - T_1 \right]\beta_{\calE_k,n}\\
& + o_p(n^{-1/2}).
\end{align*}
Note that $\beta_{\calE_k}=O(n^{-1/2})$.
Further, observe that 
$$\frac1n X\tran \diag(\nabla^2 A(X_E\beta_{E,n}) (X_{E^{(k)}}\calI_{E^{(k)},E^{(k)}}^{-1}\calI_{E^{(k)},\calE_k } - X_{\calE_k } ) - T_1=o_p(1), \text{ and }$$ 
$$\frac1{n_k} X^{(k),\intercal} \diag(\nabla^2 A(X_E^{(k)}\beta_{E,n}) (X_{E^{(k)}}^{(k)}\calI_{E^{(k)},E^{(k)}}^{-1}\calI_{E^{(k)},\calE_k} - X_{\calE_k} ) - T_1=o_p(1).$$ 
Thus, we conclude that $R_1^{(k)}-r_1^{(k)}=o_p(n^{-1/2})$.
\end{proof}

\begin{lemma}[Rate of $R^{(k)}_2 - r^{(k)}_2 $]
\label{lem: R_2 - r_2}
Let
\begin{align*}
R_2^{(k)}&=\frac1n X\tran (\nabla A(X\widehat\beta^{\Lambda,(k)}) - \nabla A(X_{E^{(k)}}\beta^*_{E^{(k)}, n} ) ),\\
r_2^{(k)}&=\frac1{n_k} X^{(k),\intercal} (\nabla A(X^{(k)}\widehat\beta^{\Lambda,(k)}) - \nabla A(X^{(k)}_{E^{(k)}}\beta^*_{E^{(k)}, n} ) ).
\end{align*}
Then
$\sqrt n(R^{(k)}_2 - r^{(k)}_2)=o_p(1) $.
\end{lemma}
\begin{proof}[Proof of Lemma~\ref{lem: R_2 - r_2}.]
Based on the assertion in Lemma~\ref{lem: property of beta*}, we have 
$$\widehat\beta^{\Lambda,(k)}_{E^{(k)}} - \beta^*_{E^{(k)}, n}=O_p(n^{-1/2}).$$
Taking a Taylor expansion of $\nabla A(X\widehat\beta^{\Lambda,(k)} )$ at $X\beta^*_{E^{(k)}, n} $ for each coordinate, we obtain
\begin{align*}
R_2^{(k)}-r_2^{(k)}&=\frac1n X\tran \left[ \diag(\nabla^2 A(X_{E^{(k)}}  \beta^{*}_{E^{(k)}, n})) X_{E^{(k)}} (\widehat\beta^{\Lambda,(k)}_{E^{(k)}}-\beta^*_{E^{(k)}, n}) + o(\|\widehat\beta^{\Lambda,(k)}_{E^{(k)}}-\beta^*_{E^{(k)}, n}\|) )
\right] - \\
&\frac{1}{n_k}X^{(k),\intercal} \left[ \diag(\nabla^2 A(X_{E^{(k)}}^{(k)} \beta^*_{E^{(k)}, n})) X_{E^{(k)}}^{(k)} (\widehat\beta^{\Lambda,(k)}_{E^{(k)}}-\beta^*_{E^{(k)}, n}) + o(\|\widehat\beta^{\Lambda,(k)}_{E^{(k)}}-\beta^*_{E^{(k)}, n}\|) )\right].
\end{align*}
%
Letting 
$$T=\EE{\frac1n X\tran \diag(\nabla^2 A(X_{E^{(k)}}\beta^*_{E^{(k)}, n})) X_{E^{(k)}}},$$
we have
$$\frac1n X\tran \diag(\nabla^2 A(X_{E^{(k)}}\beta^*_{E^{(k)}, n})) X_{E^{(k)}}  = T+o_p(1),$$ 
and 
$$\frac{1}{n_k}X^{(k),\intercal}  \diag(\nabla^2 A(X_{E^{(k)}}^{(k)} \beta^*_{E^{(k)}, n})) X_{E^{(k)}}^{(k)}=T+o_p(1).$$
Hence,
\begin{align*}
R_2^{(k)}-r_2^{(k)}&=\left[ \frac1n X\tran \diag(\nabla^2 A(X_{E^{(k)}}\beta^*_{E^{(k)}, n})) X_{E^{(k)}}  - T\right] (\widehat\beta^{\Lambda,(k)}_{E^{(k)}}-\beta^*_{E^{(k)}, n}) \\
& - \left[\frac{1}{n_k}X^{(k),\intercal}  \diag(\nabla^2 A(X_{E^{(k)}}^{(k)} \beta^*_{E^{(k)}, n})) X_{E^{(k)}}^{(k)} - T \right] (\widehat\beta^{\Lambda,(k)}_{E^{(k)}}-\beta^*_{E^{(k)}, n}) + o_p(n^{-1/2})\\
&=o_p(n^{-1/2}).
\end{align*}
\end{proof}
\indent
\end{proof}

\subsection{Proof of Proposition~\ref{prop: asymp indep}}
\label{prf: asymp indep}

\begin{proof}
It follows from Assumption~\ref{assump: glm regularity} that
\[
\sqrt n(\widehat\beta_E-\beta_{E,n})\convdis\N(\mathbf{0},\, \calI_{E,E}^{-1}),
\]
and Theorem~\ref{thm: normality of omega} proves that $\sqrt n\Omega\convdis\N(\mathbf{0},\,\Sigma_{\Omega})$.

For our remaining estimator, we note that
\begin{align*}
\sqrt n\widehat\beta_{-E}^\perp&=\frac1{\sqrt n} X_{-E}\tran(\nabla A(X_E\beta_{E,n})-Y) + \frac1{\sqrt n} X_{-E}\tran(\nabla A(X_E\widehat\beta_{E})-\nabla A(X_E\beta_{E,n}))\\
&=\frac1{\sqrt n} X_{-E}\tran(\nabla A(X_E\beta_{E,n})-Y) + \frac1{\sqrt n} X_{-E}\tran W X_E\tran(\widehat\beta_{E}-\beta_{E,n}) + o_p(1)\\
&=\frac1{\sqrt n} X_{-E}\tran(\nabla A(X_E\beta_{E,n})-Y) - \calI_{-E,E}\calI_{E,E}^{-1} \frac{1}{\sqrt n} X_E\tran(\nabla A(X_E\beta_{E,n})-Y) + o_p(1).
\end{align*}
In particular, we have
\begin{align*}
\text{Var}\begin{pmatrix}
\frac1{\sqrt n} X_{E}\tran (\nabla A(X_E\beta_{E,n})-Y) \\ \frac{1}{\sqrt n} X_{-E}\tran(\nabla A(X_E\beta_{E,n})-Y)
\end{pmatrix}=
\begin{pmatrix}
\calI_{E,E} & \calI_{E,-E}\\ \calI_{-E,E} & \calI_{-E,-E}
\end{pmatrix}.
\end{align*}
Thus, we observe that the asymptotic variance of $\sqrt n \widehat\beta_{-E}^{\perp}$ is equal to
\begin{align*}
\calI_{-E,-E}-\calI_{-E,E}\calI_{E,E}^{-1}\calI_{E,-E},
\end{align*}
and conclude that $$\sqrt n\widehat\beta_{-E}^\perp\convdis\N(\mathbf{0},\calI/\calI_{E,E}).$$
Further, it is easy see from the asymptotic representations of $\sqrt n(\widehat\beta_E-\beta_{E,n})$ and $\sqrt n\widehat\beta^\perp_{-E}$
that they are mutually independent.

Now, observe that $\sqrt n(\widehat\beta_E-\beta_{E,n})$ and $\sqrt n\widehat\beta^\perp_{-E}$ are asymptotically equivalent to sums of i.i.d. random variables $z_i$, and that $\sqrt{n}\omega^{(k)}$ assumes the form in \eqref{rep:omega}.
Thus, we can write
\begin{equation*}
 \begin{aligned}
     \text{Cov}\left(\sqrt{n}\tilde\omega^{(k)}, \sum_{i\in[n]} z_i\right)
     =\sqrt n\text{Cov}\left(\frac1n\sum_{i\in[n]}e_i - \frac{1}{n_k}\sum_{i\in\calC_k} e_i,\sum_{i\in[n]}z_i \right)
 \end{aligned}
\end{equation*}
where $e_i=x_i(\nabla A(x_{i,E}\tran \beta_{E,n}) - y_i)$.
The independence between the randomization variables and the remaining variables follows the fact that the right-hand-side in the last display is $\mathbf{0}$.
\end{proof}

\subsection{Proof of Theorem~\ref{sel:event:rep}}
\label{prf: sel:event:rep}

\begin{proof}
The proof for part $(\rom{1})$ is direct and thus omitted.

For part $(\rom{2})$, we start from writing 
$$\sqrt n\omega^{(k)} = \sqrt n\bar\omega^{(k)}+o_p(1),$$  established in Lemma \ref{lem: rand:map}, where
\begin{align*}
\sqrt n\bar\omega^{(k)}&=
\mathbb{T}^{(k)}(\sqrt n\widehat B^{(k)}, \widehat{Z}^{(k)}; \sqrt n\widehat \beta_E, \sqrt n\widehat\beta^\perp_{-E} ).
\end{align*}
Define
$$\widebar W =\begin{pmatrix}\bar\omega^{(1)} \\ \vdots \\ \bar\omega^{(K)}\end{pmatrix}=\bbT(\sqrt n\widehat B,\widehat Z; \sqrt n\widehat\beta_E, \sqrt n\widehat\beta^{\perp}_{-E} ).$$

We begin with $p_{n}$, the Lebesgue density of
$$ \left( \sqrt{n} (\widehat\beta_E-\beta_{E,n}),\; \sqrt n \widehat\beta_{-E}^\perp,\; \sqrt{n}\widebar{W} \right). $$
We then apply the change of variables
$$\sqrt{n}\widebar{W} \to (\sqrt n\widehat B, \widehat Z)$$
through the mapping
$\sqrt n\widebar W =\bbT(\sqrt n\widehat B,\widehat Z; \sqrt n\widehat\beta_E, \sqrt n\widehat\beta^{\perp}_{-E} ) $. 
Because the mapping is linear, the density for
$$ \left( \sqrt{n} (\widehat\beta_E-\beta_{E,n}),\; \sqrt n \widehat\beta_{-E}^\perp,\; \sqrt{n}\widehat B,\; \widehat Z \right)
$$
is proportional to
\begin{align*}
p_n(\sqrt n(\widehat\beta_E - \beta_{E,n}),\; \sqrt n\widehat\beta^\perp_{-E},\; \bbT(\sqrt n\widehat B,\widehat Z; \sqrt n\widehat\beta_E, \sqrt n\widehat\beta^{\perp}_{-E}) ).
\end{align*}
Now, the condition in Assumption \ref{assump: densities} allows us to replace $p_n$ by the limiting Gaussian density in Proposition~\ref{prop: asymp indep} which gives us the corresponding asymptotic density function 
\begin{equation}
\begin{aligned}
 \varphi(\sqrt n\widehat\beta_E;\sqrt n\beta_{E,n},\calI_{E,E}^{-1})  \times
\varphi(\sqrt{n}\widehat\beta_{-E}^\perp; & \mathbf{0},(\calI/\calI_{E,E})^{-1})\\
 &\times \varphi(\bbT(\sqrt n\widehat B,  \widehat Z;\sqrt n\widehat\beta_E, \sqrt n \widehat\beta_{-E}^\perp);\mathbf{0}, \Sigma_{\Omega} ).
 \end{aligned}
\label{equ: prf thm 4.3 density}
\end{equation}
Furthermore, if we condition on $\widehat Z=Z$ and ignore constants, the asymptotic likelihood can be simplified to 
\[
\varphi(\sqrt n\widehat\beta_E; \Pi \sqrt n\beta_{E,n} + \kappa; \Theta)\cdot \varphi(\sqrt n\widehat B; \Psi \sqrt n\widehat\beta_E + \tau;\Gamma),
\]
where $\Pi,\kappa,\Theta,\Psi,\tau,\Gamma$ are defined in Section~\ref{sec: algorithm}. 
See details of the simplification in Lemma~\ref{lem: matrix:simplification} below.

\begin{lemma}[Matrix simplification]
\label{lem: matrix:simplification}
The joint density of $(\sqrt n\widehat\beta_E, \sqrt n\widehat B,\widehat\beta^\perp_{-E})$ when conditioned on $\widehat Z= Z$ is equal to
\[
\varphi(\sqrt n\widehat\beta_E; \Pi \sqrt n\beta_{E,n} + \kappa; \Theta)\cdot \varphi(\sqrt n\widehat B; \Psi \sqrt n\widehat\beta_E + \tau;\Gamma)\cdot \varphi(\sqrt n\widehat\beta^\perp_{-E};\mathbf{0},(\calI/\calI_{E,E})^{-1} ).
\]
\end{lemma}
\begin{proof}[Proof of Lemma~\ref{lem: matrix:simplification}.]
Denote
\begin{align}
\bbQ_1=\begin{pmatrix}
\calI_{\cdot,E} \\ \vdots \\ \calI_{\cdot,E}
\end{pmatrix},\quad
\bbQ_2=\begin{pmatrix}
\calI_{\cdot,E^1} & \mathbf{0} & \ldots & \mathbf{0} \\ \mathbf{0} & \calI_{\cdot,E^2} & \ldots & \mathbf{0}\\  &  & \ddots & \\  \mathbf{0} & \ldots & \mathbf{0} & \calI_{\cdot,E^K}
\end{pmatrix}.
\label{equ: def Q_1 Q_2}
\end{align}
Let $r^{(k)}=\begin{pmatrix} S^{(k)} \\ \widehat{Z}^{(k)} \end{pmatrix} + \begin{pmatrix} \mathbf{0} \\ \sqrt n \widehat\beta_{-E}^\perp\end{pmatrix} $, and let $\bfr$ be the stack of $r^{(1)},\ldots,r^{(K)} $.
Observe, the mapping $\bbT$ in the proof of Theorem~\ref{sel:event:rep}  can be written as
\begin{equation}
\sqrt n\widebar W=-\bbQ_1 \sqrt n\widehat\beta_E + \bbQ_2 \sqrt n\widehat B + \bfr.
\label{rand:map:matrixform}
\end{equation}
It follows from Equation~\eqref{equ: prf thm 4.3 density} that the joint density of $(\sqrt n\widehat\beta_E,\sqrt n\widehat B)$ after conditioning on $\widehat{Z}=Z$ is proportional to
\begin{align*}
&\Exp{-\frac12 (\sqrt n\widehat\beta_E-\sqrt n\beta_{E,n})\tran \calI_{E,E} (\sqrt n\widehat\beta_E-\sqrt n\beta_{E,n})\right.\\
&\qquad\qquad\left.-\frac12 (-\bbQ_1\sqrt n\widehat\beta_E+\bbQ_2 \sqrt n\widehat B +\bfr)\tran\Sigma_{\Omega}^{-1} (-\bbQ_1\sqrt n\widehat\beta_E+\bbQ_2 \sqrt n\widehat B +\bfr)}\\
&\propto \Exp{-\frac{1}{2}(\sqrt n\widehat B)\tran \bbQ_2\tran\Sigma_{\Omega}^{-1}\bbQ_2(\sqrt n\widehat B)+ (\sqrt n\widehat B)\tran \bbQ_2\tran \Sigma_{\Omega}^{-1}(\bbQ_1\sqrt n\widehat\beta_E-\bfr)\right.\\
&\qquad\qquad\left. -\frac{1}{2}(\sqrt n\widehat\beta_E)\tran (\calI_{E,E}+\bbQ_1\tran \Sigma_\Omega^{-1}\bbQ_1 )(\sqrt n\widehat\beta_E) + (\sqrt n\widehat\beta_E)\tran(\calI_{E,E} \sqrt n\beta_{E,n}+\bbQ_1\tran \Sigma_{\Omega}^{-1}\bfr)}.
\end{align*}

For $\Gamma^{-1}=\bbQ_2\tran \Sigma_{\Omega}^{-1}\bbQ_2$, $\Psi=\Gamma \bbQ_2\tran\Sigma_{\Omega}^{-1}\bbQ_1$, $\tau=-\Gamma\bbQ_2\tran \Sigma_{\Omega}^{-1} \bfr $, we observe that this likelihood is proportional to
\begin{align*}
\varphi(\sqrt n\widehat B;\Psi\sqrt n\widehat\beta_E+\tau,\Gamma)&\cdot \Exp{-\frac{1}{2}(\sqrt n\widehat\beta_E)\tran(\calI_{E,E}+\bbQ_1\tran\Sigma_\Omega^{-1}\bbQ_1- \Psi\tran\Gamma^{-1}\Psi)(\sqrt n\widehat\beta_E)\right.\\
&\qquad\left.+(\sqrt n\widehat\beta_E)\tran(\Psi\tran\Gamma^{-1}\tau+\calI_{E,E}\sqrt n\beta_E+\bbQ_1\tran\Sigma_{\Omega}^{-1}\bfr) }.
\end{align*}
The likelihood function in the last display is proportional to
\begin{align*}
\varphi(\sqrt n\widehat\beta_E; \Pi \sqrt n\beta_E + \kappa; \Theta)\cdot \varphi(\sqrt n\widehat B; \Psi \sqrt n\widehat\beta_E + \tau;\Gamma )
\end{align*}
for $\Theta^{-1}=\calI_{E,E}+\bbQ_1\tran\Sigma_{\Omega}^{-1}\bbQ_1-\Psi\tran\Gamma^{-1}\Psi$, $\Pi=\Theta \calI_{E,E}$, $\kappa=\Theta(\Psi\tran\Gamma^{-1}\tau+\bbQ_1\tran\Sigma_{\Omega}^{-1}\bfr)$.

To further simplifying the matrices in the likelihood, we note that $\Sigma_{\Omega}^{-1}=U^{-1}\otimes \calI^{-1}$.
Therefore, we can write $\bbQ_2\tran \Sigma_{\Omega}^{-1}\bbQ_2$ in the form of $K\times K$ blocks, where the $(j,k)$-block is
\begin{align*}
\left(\frac{\rho_j\rho_k}{\rho_0}+\rho_j\delta_{j,k}\right) \calI_{E^{(j)},E^{(k)}}.
\end{align*}
Similarly, $\bbQ_2\tran\Sigma_{\Omega}^{-1}\bfr$ has $K$ blocks where the $k$-th block is
\begin{align*}
\rho_kJ_{E^{(k)}}\bfr^{(k)} +\frac{\rho_k}{\rho_0}\sum_{j=1}^K \rho_j J_{E^{(k)}} \bfr^{(j)}.
\end{align*}
Since $J_{E^{(k)}} \begin{pmatrix}0_{E}\\ \widehat\beta_{-E}^{\perp}\end{pmatrix}=0 $ and $J_{E^{(k)}}\gamma^{(j)}=g_k^{(j)}$, the above display is equal to $\rho_k g_k^{(k)}+(\rho_k/\rho_0)\sum_{j\in[K]} \rho_j g_k^{(j)} $. 
Similarly, $\bbQ_1\tran\Sigma_{\Omega}^{-1}\bfr=\sum_{k\in[K]}\frac{\rho_k}{\rho_0} J_E\gamma^{(k)} =\sum_{k\in[K]}\frac{\rho_k}{\rho_0}g^{(k)} $ and
\begin{align*}
\bbQ_1\tran\Sigma_{\Omega}^{-1}\bbQ_1=\left(\sum_{k\in[K]}\rho_k +\sum_{j,k\in[K]}\frac{\rho_j\rho_k}{\rho_0} \right)\calI_{E,E}=\frac{1-\rho_0}{\rho_0}\calI_{E,E}.
\end{align*}
Thus $\Theta^{-1}=\frac{1}{\rho_0}\calI_{E,E}-\Psi\tran \Gamma^{-1}\Psi$.

\end{proof}

\begin{remark}
\label{rmk on general aggregation}
We note that when using the union aggregation rule, the quantities $\Pi,\kappa,\Theta,\Psi,\tau,\Gamma$ do not depend on $\widehat\beta^\perp_{-E}$.
If $E^{(k)}$ is not necessarily a subset of $E$,
then
\begin{align*}
J_{E^{(k)}} \begin{pmatrix}0_{E}\\ \widehat\beta_{-E}^{\perp}\end{pmatrix}=\begin{pmatrix}
\mathbf{0}_{ E^{(k)}\cap E } \\ \widehat\beta^{\perp}_{E^{(k)}\setminus E },
\end{pmatrix}
\end{align*}
and thus
\[
J_{E^{(k)}}\bfr^{(j)} = J_{E^{(k)}} \gamma^{(j)}  + \begin{pmatrix}
\mathbf{0}_{ E^{(k)}\cap E } \\ \widehat\beta^{\perp}_{E^{(k)}\setminus E }
\end{pmatrix}. 
\]
But, we only need to re-define $g_k^{(j)}$ as 
\[
g^{(j)}_k =J_{E^{(k)}} \gamma^{(j)} + \begin{pmatrix} \mathbf{0}_{E^{(k)}\cap E } \\ \widehat\beta^{\perp}_{E^{(k)} \setminus E }  \end{pmatrix},
\]
and this change only affects $\tau$. 
For the complete procedure when using general aggregation rules, see Appendix~\ref{sec: general aggregation}.
\end{remark}

\end{proof}

\subsection{Supporting Results}
\label{appendix:supp}

\begin{lemma}
\label{lem: property of beta*}
Let $\beta^*_{E^{(k)}, n}=\calI_{E^{(k)},E^{(k)}}^{-1}\calI_{E^{(k)},E}\beta_{E,n} $. Under Assumption~\ref{assump: missed variables}, 
\begin{align*}
    X_{E^{(k)}} \beta^{*}_{E^{(k)}, n} - X_E \beta_{E,n} &=O_p(n^{-1/2}),\text { and }\\
\widehat\beta_{E^{(k)}}^{(k),\Lambda}-\beta^*_{E^{(k)}, n}&=O_p(n^{-1/2}).
\end{align*}
\end{lemma}
\begin{proof}
Define
\begin{align*}
\nabla \ell_{E^{(k)}}(\beta)=\frac{1}{\sqrt n}\sum_{i\in[n]} x_{i,E^{(k)}}(\nabla A(x_{i,E^{(k)}}\tran \beta_{E^{(k)}}) - y_i )
\end{align*}
as the gradient of log-likelihood w.r.t. $\beta_{E^{(k)}}$. Then by the optimality of Lasso solution, we have
\begin{align*}
\nabla \ell_{E^{(k)}}(\widehat\beta^{\Lambda,(k)}_{E^{(k)}})=-\Lambda_{E^{(k)}} s_k.
\end{align*}
Taking a Taylor expansion of $\nabla\ell_{E^{(k)}}(\widehat\beta^{(k)}_{E^{(k)}})$ at $\beta^*_{E^{(k)}, n}$, we get
\begin{equation}
\begin{aligned}
-\Lambda_{E^{(k)}} s_k &=\nabla \ell_{E^{(k)}}(\widehat\beta^{\Lambda,(k)}_{E^{(k)}})=\nabla \ell_{E^{(k)}}(\beta^*_{E^{(k)}, n}) + \nabla^2\ell_{E^{(k)}}(\beta^*_{E^{(k)}, n})(\widehat\beta^{\Lambda,(k)}_{E^{(k)}} - \beta^*_{E^{(k)}, n} )\\
&+ o(\|\widehat\beta^{\Lambda,(k)}_{E^{(k)}} - \beta^*_{E^{(k)}}\|_2 ).
\label{equ: expansion at beta*}
\end{aligned}
\end{equation}
Note that
\begin{align*}
X_{E^{(k)}} \beta^{*}_{E^{(k)}} -X_E  \beta_E &=(X_{E^{(k)}}\calI_{E^{(k)},E^{(k)}}^{-1}\calI_{E^{(k)},E} - X_E )\beta_E\\
&=\left(
X_{E^{(k)}}\calI_{E^{(k)},E^{(k)}}^{-1}\calI_{E^{(k)},E\setminus E^{(k)} } - X_{E\setminus E^{(k)}} \right) \beta_{E\setminus E^{(k)}}+\\
&\qquad \left(X_{E^{(k)}}\calI_{E^{(k)},E^{(k)}}^{-1}\calI_{E^{(k)},E\cap E^{(k)} } - X_{E\cap E^{(k)}}\right)\beta_{E\cap E^{(k)} }\\
&=\left(
X_{E^{(k)}}\calI_{E^{(k)},E^{(k)}}^{-1}\calI_{E^{(k)},E\setminus E^{(k)} } - X_{E\setminus E^{(k)}} \right) \beta_{E\setminus E^{(k)}}
\end{align*}

By Assumption~\ref{assump: missed variables}, if $j\in(E\setminus E^{(k)})\setminus \calE_k$,
\[
X_{E^{(k)}} \calI_{E^{(k)},E^{(k)}}^{-1}\calI_{E^{(k)},j} - X_{j}=0.
\]
Since $\beta_{\calE_k}=O(n^{-1/2})$, we have
\begin{align*}
    X_{E^{(k)}} \beta^{*}_{E^{(k)},n} - X_E \beta_E &=(X_{E^{(k)}} \calI_{E^{(k)},E^{(k)}}^{-1}\calI_{E^{(k)},\calE_k} - X_{\calE_k}) \beta_{\calE_k}=O_p(n^{-1/2})
\end{align*}

Note that
\begin{align*}
\nabla\ell_{E^{(k)}}(\beta^*_{E^{(k)}})&=\frac{1}{\sqrt n}\sum_{i\in[n]}x_{i,E^{(k)}}(\nabla A(x_{i,E_k}\tran\beta^*_{E^{(k)}, n}) - y_i )\\
&=\frac{1}{\sqrt n}\sum_{i\in[n]}x_{i,E^{(k)}}(\nabla A(x_{i,E_k}\tran\beta^*_{E^{(k)}, n}) - \nabla A(x_{i,E}\tran \beta_{E,n}) ) \\
&+ \frac{1}{\sqrt n}\sum_{i\in[n]}x_{i,E^{(k)}}(\nabla A(x_{i,E}\tran\beta_{E,n}) - y_i )\\
&=\frac{1}{\sqrt n} X_{E^{(k)}}\tran \diag(\nabla^2 A(X_E\beta_E)) (X_{E^{(k)}} \beta^{*}_{E^{(k)},n} - X_E \beta_E )\\
&\quad +\frac{1}{\sqrt n}\sum_{i\in[n]}x_{i,E^{(k)}}(\nabla A(x_{i,E}\tran\beta_{E, n}) - y_i )+ o_p(1)\\
&=\sqrt n O(n^{-1/2}) + O_p(1) + o_p(1)\\
&=O_p(1).
\end{align*}

Moreover,
\begin{align*}
\frac{1}{\sqrt n}\nabla^2\ell_{E^{(k)}}(\beta^*_{E^{(k)}})&=\frac{1}{n}X_{E^{(k)}} \tran \diag(\nabla^2 A(X_{E^{(k)}}\beta^*_{E^{(k)}} )) X_{E^{(k)}}\\
&=\frac1n X_{E^{(k)}}\tran \diag(\nabla^2 A(X_E \beta_E)) X_{E^{(k)}} + o_p(1)\\
&=\calI_{E^{(k)},E^{(k)}} + o_p(1).
\end{align*}
Substituting into Equation~\eqref{equ: expansion at beta*} gives
\begin{align*}
\sqrt n(\widehat\beta^{\Lambda,(k)}_{E^{(k)}} - \beta^*_{E^{(k)}, n}) = O_p(1).
\end{align*}
\end{proof}

\begin{lemma}\label{lem: rand:map}
The randomization variables have the expression
\begin{align*}
\sqrt n\omega^{(k)} = \sqrt n\bar\omega^{(k)} + o_p(1),
\end{align*}
where
\begin{align*}
\sqrt n\bar\omega^{(k)}&=
\mathbb{T}^{(k)}(\sqrt n\widehat B^{(k)}, \widehat{Z}^{(k)}; \sqrt n\widehat \beta_E, \sqrt n\widehat\beta^\perp_{-E} )\\
&=\calI_{\cdot,E^{(k)}} \sqrt n\widehat B^{(k)} - \calI_{\cdot, E}\sqrt n\widehat\beta_E + \Lambda \begin{pmatrix} S^{(k)} \\ \widehat{Z}^{(k)} \end{pmatrix} + \begin{pmatrix} \mathbf{0} \\ \sqrt n \widehat\beta_{-E}^\perp \end{pmatrix}.
\end{align*}
\end{lemma}

\begin{proof}
The K.K.T. conditions of stationarity for the Lasso on machine $k$ are summarized by
\begin{equation}
\begin{aligned}
\sqrt n\omega^{(k)}&=  \frac{1}{\sqrt n}X\tran(\nabla A(X \widehat\beta^{(k),\Lambda}) - Y) + \widehat\gamma^{(k)}\\
&=\frac1{\sqrt n} X\tran( \nabla A(X_E \widehat\beta_E) - Y  + \nabla A(X_{E^{(k)}} \widehat B^{(k)}) -  \nabla A(X_{E}\widehat\beta_E) ) + \widehat\gamma^{(k)}\\
&=\begin{pmatrix}
\mathbf{0} \\ \sqrt n \widehat\beta_{-E}^\perp
\end{pmatrix} +
\frac{1}{\sqrt n}X\tran \diag(\nabla^2 A(X_E\beta_{E,n})) X_{E^{(k)}} \widehat B^{(k)}  \\ 
&\qquad\qquad -\frac{1}{\sqrt n}X\tran \diag(\nabla^2 A(X_E\beta_{E,n})) X_E\widehat\beta_E +\widehat\gamma^{(k)}+o_p(1)\\
&=\calI_{\cdot,E^{(k)}} \sqrt n\widehat B^{(k)} - \calI_{\cdot,E} \sqrt n\widehat\beta_E + \Lambda \begin{pmatrix} S^{(k)} \\ \widehat{Z}^{(k)} \end{pmatrix} + \begin{pmatrix}
\mathbf{0} \\ \sqrt n \widehat\beta_{-E}^\perp
\end{pmatrix} \\
&= \sqrt{n} \bar{w}^{(k)} + o_p(1).
\end{aligned}
\label{rand:map}
\end{equation}

\end{proof}

\section{Proofs for Section~\ref{sec: asymp sel likelihood}}

\subsection{Proof of Theorem \ref{thm:consistency}}
Before proving Theorem \ref{thm:consistency}, we provide a supporting result in Lemma \ref{MDP:V}.

First, let $s^{(k)}= \begin{pmatrix} g_k^{(k)} \\ \mathbf{0} \end{pmatrix} \in \mathbb{R}^p$ be the active components of the subgradient vector $\gamma^{(k)}$ padded with a vector of all zeros.
Now, let $\bfs \in \mathbb{R}^{pK}$ be formed by stacking the vectors $s^{(k)}$ for $k\in [K]$, and let $G \in \mathbb{R}^{\bar{d}}$ be formed by stacking $g_k^{(k)}$ for $k\in [K]$.
Recall that $\widebar{W}$ is formed by stacking the vectors $\bar{\omega}^{(k)}$, $k\in [K]$ that we previously defined in Lemma \ref{lem: rand:map}. 
Suppose that 
$$\sqrt{n}\bar\Omega = \sqrt{n}\widebar{W} - \bfs.$$

\begin{lemma}
Define
$$V_n = \begin{pmatrix}\sqrt{n}\widehat{\beta}_E \\ \sqrt{n}\widehat{\beta}^\perp_{-E} \\ \sqrt{n}\bar\Omega \end{pmatrix}.$$
Let $\mathcal{S}$ be a convex subset of $\mathbb{R}^{p(K+1)}$.
Under the conditions in Assumption \ref{as:1} and Assumption \ref{as:2}, we have
\begin{equation*}
\begin{aligned} 
& \displaystyle\lim_{n\to \infty} -\dfrac{1}{a_n^2} \log \mathbb{P}\left[ \frac{1}{a_n} V_n \in \mathcal{S}\right]=\inf_{b, b^{\perp}, \omega \in \mathcal{S}}R(b, b^{\perp}, \omega),
\end{aligned}
\end{equation*}
where 
$$R(b, b^{\perp}, \omega)=  \Bigg\{  \frac{1}{2}(b - \beta_E)\tran \calI_{E,E}(b - \beta_E) + \frac{1}{2}(b^\perp)\tran(\calI\setminus \calI_{E,E})^{-1}b^\perp + \frac{1}{2}\omega\tran \Sigma_{\Omega}^{-1}\ \omega\Bigg\}.$$
\label{MDP:V}
\end{lemma}

\begin{proof}
Based on the condition in Assumption \ref{as:2}, we have
\begin{equation*}
\begin{aligned} 
& \displaystyle\lim_{n\to \infty} \dfrac{1}{a_n^2} \left\{\log \mathbb{P}\left[ \frac{1}{a_n} V_n \in \mathcal{S}\right]- \log \mathbb{P}\left[ \frac{1}{a_n} \begin{pmatrix}\sqrt{n}\widehat{\beta}_E \\ \sqrt{n}\widehat{\beta}^\perp_{-E} \\ \sqrt{n}\Omega \end{pmatrix} \in \mathcal{S}\right] \right\} =0.
\end{aligned}
\end{equation*}
Because,
$$
\sqrt{n}\begin{pmatrix}\widehat{\beta}_E \\ \widehat{\beta}^\perp_{-E} \\ \Omega \end{pmatrix} = \sqrt{n} \bar{E}_n +R_n,
$$
we have the following large-deviation limit 
$$
\displaystyle\lim_{n\to \infty} -\dfrac{1}{a_n^2} \log \mathbb{P}\left[ \frac{1}{a_n} \begin{pmatrix}\sqrt{n}\widehat{\beta}_E \\ \sqrt{n}\widehat{\beta}^\perp_{-E} \\ \sqrt{n}\Omega \end{pmatrix}  \in \mathcal{S}\right]=\inf_{b, b^{\perp}, \omega \in \mathcal{S}}R(b, b^{\perp}, \omega)
$$
under Assumption \ref{as:1}.
\end{proof}

Now we are ready to prove Theorem \ref{thm:consistency}.
\begin{proof}
Let $\bbQ_1$, $\bbQ_2$ be defined according to \eqref{equ: def Q_1 Q_2}. 
Note that
$$\bfr =\bbQ_3 \widehat{Z} + \bbQ_4 \sqrt{n}\widehat\beta^\perp_{-E} +\bfs$$
for fixed matrices $\bbQ_3$ and $\bbQ_4$.
Thus, for  
$$h(B,Z; b, b^\perp)= \bbQ_2 B + \bbQ_3 Z -\bbQ_1 b + \bbQ_4 b^\perp,$$
we have
$$\sqrt{n}\bar\Omega =h(\sqrt{n}\widehat{B},\widehat{Z}; \sqrt{n}\widehat{\beta}_E, \sqrt{n}\widehat{\beta}^{\perp}_{-E}).$$
Denote by $U_n$ the vector
$$
\begin{pmatrix}
\sqrt{n}\widehat{\beta}_E\\ 
\sqrt{n}\widehat{\beta}^\perp_{-E}\\ 
\sqrt{n} \widehat{B}\\ 
\widehat{Z} 
\end{pmatrix},$$
and consider $V_n$ which we defined in Lemma \ref{MDP:V}.
Observe that 
\begin{equation}
\begin{aligned}
\label{CoV:LDP}
V_n &=  \begin{pmatrix}[1.5]\sqrt{n}\widehat{\beta}_E \\ \sqrt{n}\widehat{\beta}^\perp_{-E} \\ h\left(\sqrt{n}\widehat{B},\widehat{Z}; \sqrt{n}\widehat{\beta}_E, \sqrt{n}\widehat{\beta}^{\perp}_{-E}\right) \end{pmatrix}.
\end{aligned}
\end{equation}
Let the expression on the right-hand side display be $\mathbb{H}(U_n)$, where $\mathbb{H}:\mathbb{R}^{p(K+1)}\to \mathbb{R}^{p(K+1)}$ is a bijective mapping.
At last, noting that 
$\tau = \bbL_1 G + \bbL_0$, and $\kappa= \bbM_1 G + \bbM_0$,
we define $\bar{\tau}= \frac{1}{a_n}(\tau - \bbL_1 G)$, and $\bar{\kappa}= \frac{1}{a_n}(\kappa- \bbM_1 G)$, and let $\frac{1}{a_n} Z = \zeta$.

Applying the contraction principle for large-deviation limits together with Lemma \ref{MDP:V}, we observe that the vector 
$$\frac{1}{a_n} U_n= \mathbb{H}^{-1}\left(\frac{1}{a_n}V_n\right)$$
 satisfies a large deviation principle with the rate function $R\circ \mathbb{H}$.
Thus, it follows that
\begin{equation*}
\begin{aligned}
&\lim_{n\to \infty} -\dfrac{1}{a_n^2}\log \mathbb{P}\left[ \frac{\sqrt{n}}{a_n} \widehat{B}\in \mathcal{O}_S \ \Big\lvert \ \frac{1}{a_n}\widehat{Z}= \zeta\right] \\
&=\inf_{b, b^{\perp}, \bbD B > \mathbf{0}} \left\{R\circ \mathbb{H}(b, b^{\perp}, B, \zeta)-\inf_{b', b'^{\perp}, \bbD B'>0} R\circ \mathbb{H}(b', b'^{\perp}, B', \zeta)\right\} \\
&= \inf_{b, \bbD B > \mathbf{0}} \frac{1}{2}(b- \Pi \beta_E - \bar{\kappa})\tran \Theta^{-1}(b- \Pi \beta_E - \bar{\kappa}) + \frac{1}{2}(B- \Psi b -\bar\tau)\tran \Gamma^{-1}(B- \Psi b -\bar\tau).
\end{aligned}
\end{equation*}

To conclude the claim, we observe that
\begin{equation*}
\begin{aligned}
& \lim_{n\to \infty}  \inf_{b, B} \Bigg\{\frac{1}{2}\left(b- \Pi \beta_E - \frac{1}{a_n}\kappa\right)\tran \Theta^{-1}\left(b- \Pi \beta_E - \frac{1}{a_n}\kappa\right)   \\
&\;\;\;\;\;\;\;\;\;\;\;\;\;\;\;\;\;\;\;\;\;\;\;\;\;\;\;\;\;\;+  \frac{1}{2}\left(B- \Psi b -\frac{1}{a_n}\tau\right)\tran \Gamma^{-1}\left(B- \Psi b -\frac{1}{a_n}\tau\right) + \frac{1}{a_n^2}\text{Barr}_{\mathcal{O}}(a_n B)\Bigg\}\\
&= \inf_{b, \bbD B> \mathbf{0}} \frac{1}{2}(b- \Pi\beta_E - \bar{\kappa})\tran \Theta^{-1}(b- \Pi \beta_E - \bar{\kappa}) + \frac{1}{2}(B- \Psi b -\bar\tau)\tran \Gamma^{-1}(B- \Psi b -\bar\tau).
\end{aligned}
\end{equation*}
This is because the sequence of convex objectives in the left-hand side display converge (in a pointwise sense) to the convex objective on the right-hand side display which has a unique minimum.
\end{proof}

\subsection{Proof of Theorem \ref{est:eqns}}

\begin{proof}
Observe, the approximate selective likelihood is equal to
$$
(\sqrt{n}\widehat{\beta}_E)\tran \Theta^{-1}(\sqrt{n}\Pi\beta_{E,n} + \kappa) -Q_n^*\left(\Theta^{-1}(\sqrt{n}\Pi\beta_{E,n} + \kappa)\right),
$$
where 
\begin{equation}
Q_n^*(\alpha) = \sup_{v} (\sqrt{n}v)\tran \alpha - Q_n(\sqrt{n}v)
\label{sup:opt}
\end{equation}
and 
\begin{equation*}
\begin{aligned}
Q_n(\sqrt{n}v) &= \frac{1}{2}  (\sqrt{n} v)\tran \Theta^{-1}\sqrt{n} v +\inf_{V}\Big\{ \frac{1}{2}\left(\sqrt{n} V- \sqrt{n}\Psi  v -\tau\right)\tran \Gamma^{-1}\left(\sqrt{n}V- \sqrt{n}\Psi v -\tau\right) \\
&\;\;\;\;\;\;\;\;\;\;\;\;\;\;\;\;\;\;\;\;\;\;\;\;\;\;\;\;\;\;\;\;\;\;\;\;\;\;\;\;\;\;\;\;\;\;\;\;+ \text{Barr}_{\mathcal{O}}(\sqrt{n} V)\Big\}.
\end{aligned}
\end{equation*}
The score, based on the approximate selective likelihood, is equal to
\begin{equation*}
\begin{aligned}
&\sqrt{n}\Pi\tran\Theta^{-1}\left(\sqrt{n}\widehat\beta_E - \nabla Q_n^*\left(\Theta^{-1}(\sqrt{n}\Pi\beta_{E,n} + \kappa)\right)\right).
\end{aligned}
\end{equation*}
Thus, the selective MLE is given by
\begin{equation*}
\begin{aligned}
\Theta^{-1}(\sqrt{n}\Pi\widehat\beta^{(S)}_{E,n} + \kappa)&= (\nabla Q_n^*)^{-1}(\sqrt{n}\widehat\beta_E)\\
&= \nabla Q_n(\sqrt{n}\widehat\beta_E)\\
&= \Theta^{-1}\sqrt{n}\widehat\beta_E - \Psi\tran\Gamma^{-1}\left(\sqrt{n}\widehat{V}_{\widehat\beta_E}^{\star}- \sqrt{n}\Psi \widehat\beta_E -\tau\right).
\end{aligned}
\end{equation*}
That is, 
\begin{equation*}
\begin{aligned}
\sqrt{n}\widehat\beta^{(S)}_E &= \sqrt{n} \Pi^{-1}\widehat\beta_E - \Pi^{-1}\kappa + \Pi^{-1}\Theta\Psi\tran \Gamma^{-1}\left(\Psi \sqrt{n}\widehat\beta_E +\tau -\sqrt{n}\widehat{V}_{\widehat\beta_E}^{\star}\right)\\
&= \sqrt{n} \Pi^{-1}\widehat\beta_E - \Pi^{-1}\kappa + \widehat\calI_{E,E}^{-1}\Psi\tran \Theta^{-1}\left(\Psi \sqrt{n}\widehat\beta_E +\tau -\sqrt{n}\widehat{V}_{\widehat\beta_E}^{\star}\right).
\end{aligned}
\end{equation*}
Let $v^*$ be the solution of \eqref{sup:opt} when $\alpha= \sqrt{n}\Pi\widehat\beta^{(S)}_E + \kappa$.
The selective obs-FI matrix, derived from the curvature of the approximate selective likelihood, is given by
\begin{equation*}
\begin{aligned}
\widehat\calI^{(S)}_{E,E} &= \Pi\tran \Theta^{-1}\nabla^2Q_n^*\left(\sqrt{n}\Pi\widehat\beta^{(S)}_E + \kappa\right) \Theta^{-1}\Pi\\
&= \Pi\tran \Theta^{-1}\left(\nabla^2 Q_n\left(\sqrt{n}v^*\right) \right)^{-1}\Theta^{-1}\Pi\\
&= \Pi\tran \Theta^{-1}\left(\Theta^{-1}+ \Psi\tran \Gamma^{-1}\Psi - \Psi\tran \Gamma^{-1} \left(\Gamma^{-1}  + \nabla^2\text{\normalfont Barr}_{}\left(\sqrt{n}\widehat{V}_{\widehat\beta_E}^{\star}\right)\right)^{-1}\Gamma^{-1} \Psi\right)^{-1} \Theta^{-1}\Pi\\
&= \widehat\calI_{E,E}\left(\Theta^{-1}+ \Psi\tran \Gamma^{-1}\Psi - \Psi\tran \Gamma^{-1} \left(\Gamma^{-1}  + \nabla^2\text{Barr}_{}\left(\sqrt{n}\widehat{V}_{\widehat\beta_E}^{\star}\right)\right)^{-1}\Gamma^{-1} \Psi\right)^{-1}  \widehat\calI_{E,E}.
\end{aligned}
\end{equation*}
\end{proof}

\section{Sampling subsets with replacement}
\label{sec: with replacement}

\begin{lemma}
If the $K$ subsets $D^{(1)},\ldots,D^{(K)}$ are independent random samples of the dataset $D$ (rather than disjoint partitions) and each subset has size $[\rho n]$, then Theorem~\ref{thm: normality of omega} holds with
\[
\Sigma_{\Omega}=\frac{1-\rho}{\rho}I_K\otimes \calI.
\]
\label{prop: w r}
\end{lemma}
\begin{proof}
It suffices to prove that for $j\neq k$ $\Cov{\sqrt n\omega^{(j)},\sqrt n\omega^{(k)} }\to0$.
Following the proof \ref{prf: normality of omega}, we only need to show
\begin{align*}
\Cov{\frac1{\sqrt n}\sum_{i=1}^n e_i - \frac{1}{\sqrt n\rho}\sum_{i\in\calC^{(j)}}e_i,\;  \frac1{\sqrt n}\sum_{i=1}^n e_i - \frac{1}{\sqrt n\rho}\sum_{i\in\calC^{(k)}}e_i }\to 0.
\end{align*}
Let $\calC^{(j)}$ denote the index set of $D^{(j)}$.
Since $\EE{\frac1{\sqrt n}\sum_{i=1}^n e_i - \frac{1}{\sqrt n\rho}\sum_{i\in\calC^{(j)}}e_i\mid \calC^{(j)} }=0$, it remains to show that
\begin{align*}
\EE{\Cov{\frac1{\sqrt n}\sum_{i=1}^n e_i - \frac{1}{\sqrt n\rho}\sum_{i\in\calC^{(j)}}e_i,\;  \frac1{\sqrt n}\sum_{i=1}^n e_i - \frac{1}{\sqrt n\rho}\sum_{i\in\calC^{(k)}}e_i \mid \calC^{(j)},\calC^{(k)} } }=0.
\end{align*}
Note that
\begin{align*}
&\Cov{\frac1{\sqrt n}\sum_{i=1}^n e_i - \frac{1}{\sqrt n\rho}\sum_{i\in\calC^{(j)}}e_i,\;  \frac1{\sqrt n}\sum_{i=1}^n e_i - \frac{1}{\sqrt n\rho}\sum_{i\in\calC^{(k)}}e_i \mid \calC^{(j)},\calC^{(k)} }\\
&\qquad=\Cov{e_i}-\Cov{e_i}-\Cov{e_i}+|\calC^{(j)}\cap \calC^{(k)} |\frac{1}{n\rho^2}\Cov{e_i}\\
&\qquad=(|\calC^{(j)}\cap \calC^{(k)} |\frac{1}{n\rho^2}-1 )\Cov{e_i}.
\end{align*}
The proof is completed by the fact that $\EE{|\calC^{(j)}\cap \calC^{(k)} |}=n\rho^2$.
\end{proof}

In this setting, the matrices $\Gamma,\Theta,\Psi,\tau,\nu,\Pi,\kappa$ are similarly found by Theorem~\ref{sel:event:rep} and its proof with $\Sigma_\Omega=\frac{1-\rho}{\rho}I_K\otimes \calI $. So now $\Sigma_{\Omega}^{-1}=\frac{\rho}{1-\rho}I_K\otimes \calI^{-1} $. Then 
\begin{align*}
\{\Gamma^{-1}\}_{j,k}=\{\bbQ_2\Sigma_\Omega^{-1}\bbQ_2\}_{j,k}=\frac{\rho}{1-\rho}\calI_{E^{(j)},E^{(j)}}
\end{align*}
if $j=k$ and $0$ otherwise. Other matrices are similarly computed.

\section{Selective inference with general aggregation rules}
\label{sec: general aggregation}

\subsection{Algorithm}
In the main manuscript, we focused on the union aggregation rule, i.e., the final model $E$ is the union of selected variables in the base models $E^{(k)},1\leq k\leq K$. 
We show that with a slight modification, our procedure can be adapted to accommodate other aggregation rules.

The new procedure is summarized in Algorithm~\ref{algo general}.
We note that the procedure remains almost the same, except that 
\begin{align}
g^{(j)}_k =J_{E^{(k)}} \gamma^{(j)} + \begin{pmatrix} \mathbf{0}_{E^{(k)}\cap E } \\ \widehat\beta^{\perp}_{E^{(k)} \setminus E }  \end{pmatrix},
\label{equ: gjk general}
\end{align}
where
\[
\widehat\beta^\perp = \frac1n X\tran(Y-X\widehat\beta_E ).
\]
If there exists a variable that is selected by machine $k$ but is not selected in the final model $E$, then we must compensate the subgradients by the correlation between that variable and the residual vector for a more general aggregation rule. 

To compute the vector $g_k^{(j)}$, we note that the central machine requires $\widehat\beta^{\perp}_{E^u\setminus E } $, where $E^u=\cup_{k\in[K]} E^{(k)}$ is the union of the base models.
Thus, each local machine must send
$$\widehat\beta^{\perp,(k) }_{E^{u}\setminus E }=X^{(k),\intercal}_{E^u\setminus E}(Y^{(k)} - X^{(k)}\widehat\beta_E )$$
to the central machine. 
Because this quantity depends on the MLE $\widehat\beta_E$, which is computed on the central machine, the central machine must first send $\widehat\beta_E$ to the local machines. 
Our modified procedure in Algorithm~\ref{algo general}, therefore, involves two more exchanges between the central machine and local machines: (1) the central machine sends $\widehat\beta_E$ to local machines; (2) local machines send $\widehat\beta^{\perp,(k) }_{E^{u}\setminus E }$ to the central machine. 
In comparison with Algorithm~\ref{algo}, the communication cost is $|E^u\cup E|$ per local machine.
Note that this cost is comparable to the overall cost of order $O(|E|^2)$, as long as $|E^u|$ is about the same order as $|E|$.

Of course, the modified $g^{(j)}_k$ in \eqref{equ: gjk general} change some matrices in the optimization that we solve for approximately-valid selective inference.
Theoretically, our selective likelihood is now obtained by conditioning further on $\widehat\beta^\perp_{E^u\setminus E}$ besides the information from the subgradient vectors.

\begin{algorithm}
\setstretch{1.2}
\caption{General aggregation rules.}
\label{algo general}
\SetKwInOut{Input}{Input}
\SetKwInOut{Output}{Output}
\vspace{2mm}

\algorithmicrequire{ \textbf{1:} Variable Selection at Local Machines}

Machine $k$ solves \eqref{GL:regression} and sends $E^{(k)} =\text{Support}(\widehat\beta^{\Lambda,(k)})$ to the central machine.

\vspace{2mm}
\algorithmicrequire{ \textbf{2:} Modeling with selected predictors}

Central Machine aggregates $E^{(k)}$ to get the final model $E$ and forms the selected GLM in \eqref{sel:model}.

\vspace{2mm}
\algorithmicrequire{ \textbf{3:} Communication with Central Machine} 

\hspace*{0.2cm} Exchange 1: Central machine sends the set $E$ as well as $E^{u}=\cup_{k\in[K]}E^{(k)} $ to the local machines.

\hspace*{0.2cm}  Exchange 2: Local machines send back the following information
\begin{align*}
&\text{local estimators: } \  \widehat\beta_{E}^{(k)}, \; \widehat\calI_{E,E}^{(k)};
&\text{subgradient at $\widehat\beta^{\Lambda,(k)}$: } \ \gamma^{(k)}_{E^u}.
\end{align*}

\hspace*{0.2cm}  Exchange 3: Central Machine computes the MLE $\widehat\beta_E$ and sends to local machines.

\hspace*{0.2cm}  Exchange 4: Local machines compute $\widehat\beta^{\perp,(k)}_{E^u\setminus E}=X^{(k),\intercal}_{E^u\setminus E}(Y^{(k)} - X^{(k)} \widehat\beta_E ) $ and send to Central Machine.

\algorithmicrequire{ \textbf{4:} Selective Inference at Central Machine}

\begin{enumerate}[label=(\Alph*),leftmargin=1.1cm]
\item Compute $\widehat\beta^\perp_{E^u\setminus E} =\frac1n X^{(0),\intercal}_{E^u\setminus E}(Y^{(0)} - X^{(0)} \widehat\beta_E ) + \frac1n \sum_{k=1}^K \widehat\beta^{\perp,(k)}_{E^u\setminus E}  $

\item Compute $\widehat\Gamma,\widehat\Psi,\widehat\tau,\widehat\Theta,\widehat\Pi,\widehat\kappa $ as defined in Section~\ref{sec: algorithm}, with $g_j^{(k)}$ defined according to Equation~\eqref{equ: gjk general}.

\item Remaining steps match with Algorithm~\ref{algo}.
\end{enumerate}

\end{algorithm}

\subsection{Experiments}

We illustrate the performance of Algorithm \ref{algo general}
on simulated data.
In the following experiment, we consider the same setting as that in Section \ref{sec: experiments} with prespecified groups of correlated predictors.
More specifically, we consider $20$ groups of predictors with size $5$; distinct groups of predictors are uncorrelated, while all pairs of predictors within the same group have correlation equal to $0.9$.
As before, we assume there are $5$ non-zero coefficients $\beta_j$ and these nonzero coefficients are present in $5$ different groups.

Suppose that 
\[
G= \bigcup_{k\in[K]}\{G_j:\, j\in E^{(k)} \},
\]
i.e., $G$ contains groups which have at least one predictor selected by at least one of the $K$ local machines.
Our final aggregated model is formed by randomly picking one predictor from each of the selected groups (in $G$) with highly correlated predictors.

The results of our experiment are shown in Figure~\ref{fig: group shaving}.
 We see that our proposed method achieves the desired coverage probability.
 Similar patterns hold up for the lengths and power of our confidence intervals as was already noted for the previous aggregation rule.

\begin{figure}
\centering
\begin{subfigure}{.8\textwidth}
\centering
\includegraphics[width=\textwidth]{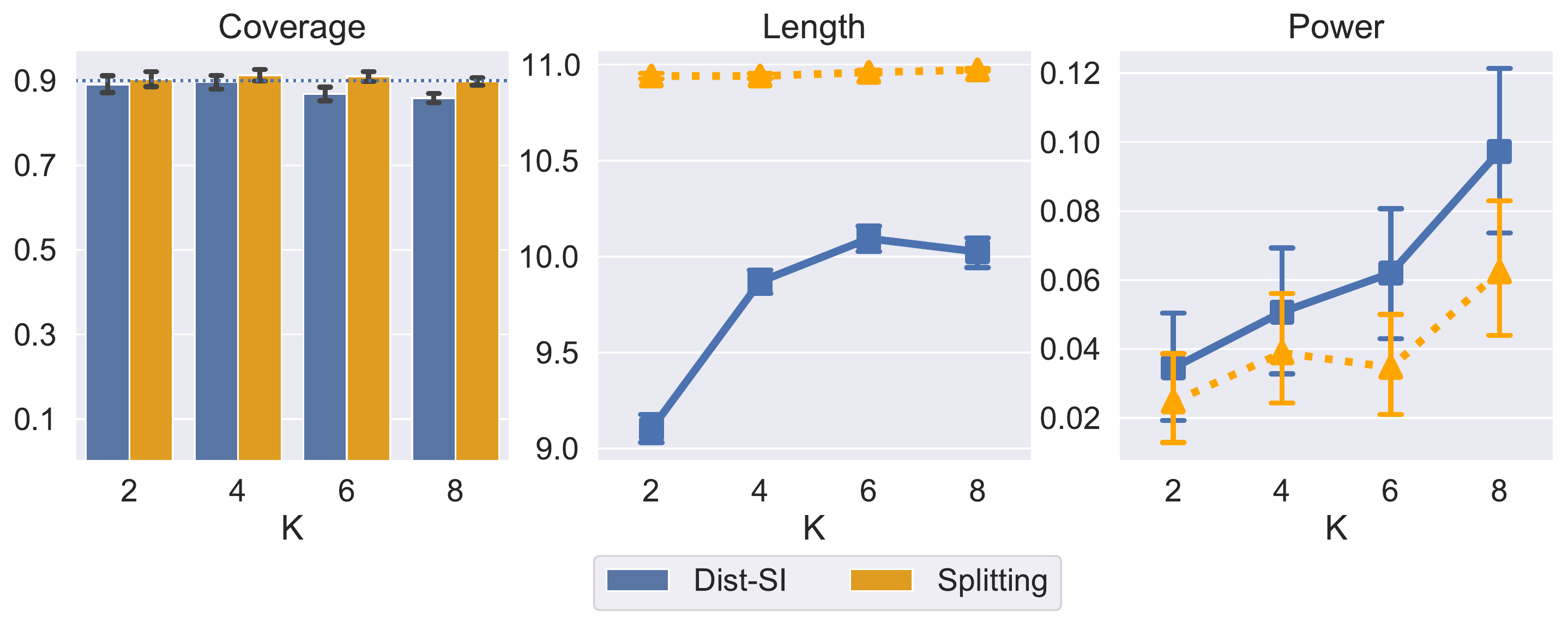}
\caption{Varying $K$. Each local machine has $[8000/K]$ data points for each $K$. 
}
\label{fig: linear vary K group}
\end{subfigure}
\begin{subfigure}{.8\textwidth}
\centering
\includegraphics[width=\textwidth]{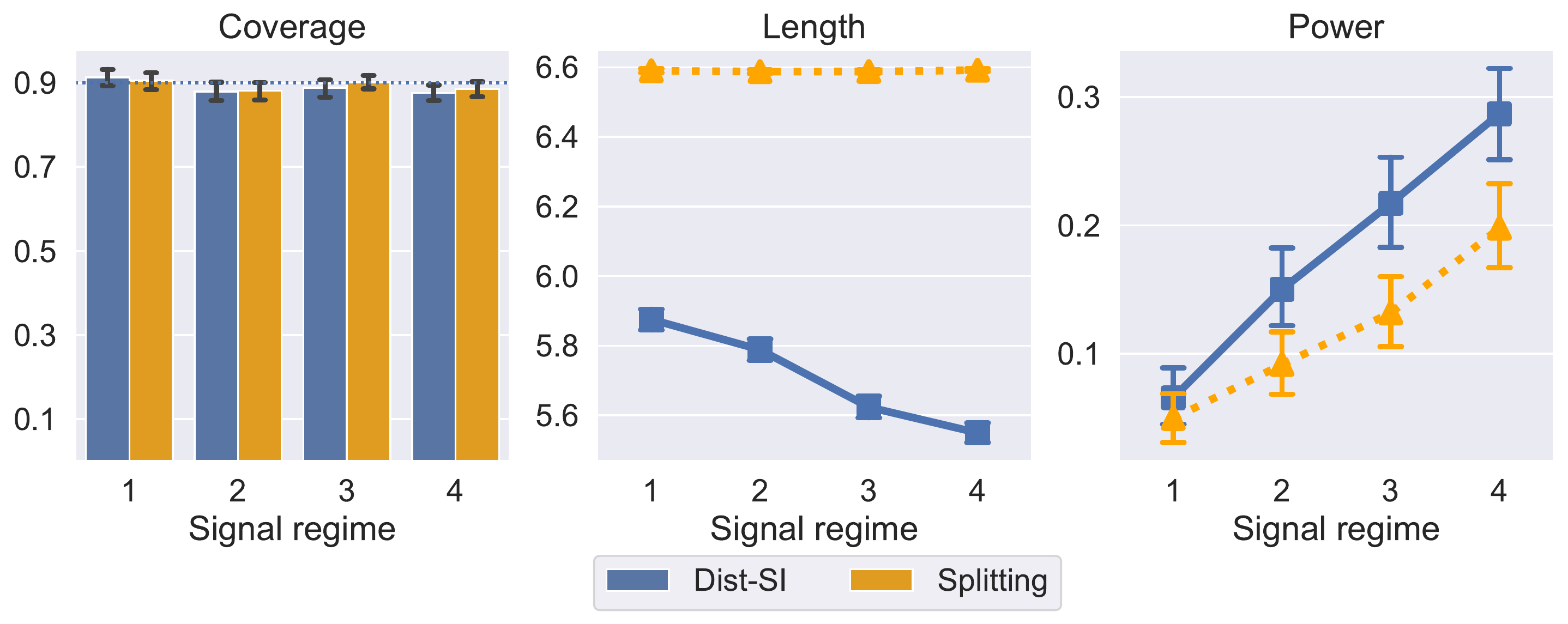}
\caption{Varying signal strength. The nonzero $\beta_j$ equals $\pm\sqrt{2c\log p}$ with random signs for $c=0.3,0.5,0.7,0.9$ in the four signal regimes.}
\label{fig: linear vary signal group}
\end{subfigure}
\begin{subfigure}{.8\textwidth}
\centering
\includegraphics[width=\textwidth]{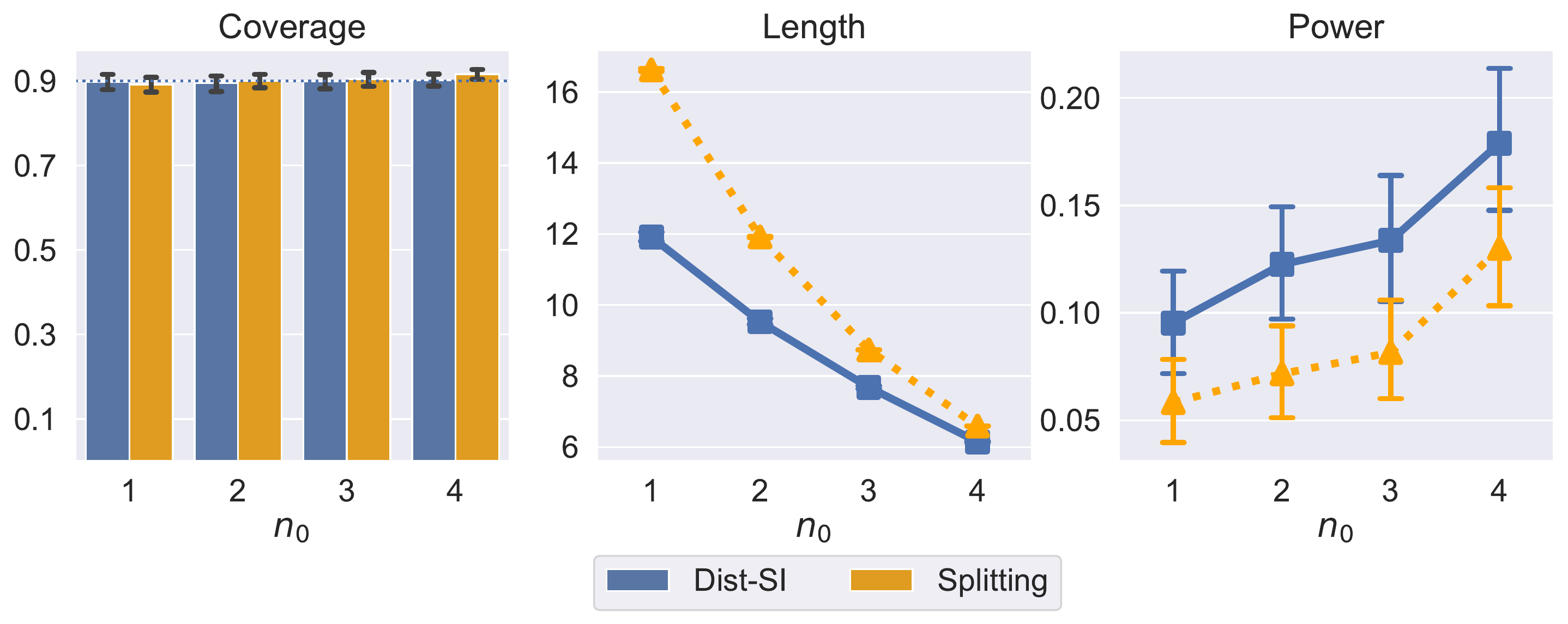}
\caption{Varying $n_0$, the sample size in the central machine.}
\label{fig: linear vary n0 group}
\end{subfigure}
\caption{Results for the grouped aggregation rule}
\label{fig: group shaving}
\end{figure}

\vskip 0.2in
\bibliographystyle{apalike}
\bibliography{paper.bbl}

\end{document}